\documentclass[a4paper,twoside,11pt]{article}
\usepackage[margin=2.54cm]{geometry}
\usepackage[affil-it]{authblk}
\usepackage{amsmath,amsfonts,amsthm,amssymb}
\usepackage{graphicx}
\usepackage{subfig}
\usepackage{enumerate}
\usepackage{paralist}
\usepackage{xspace}
\usepackage{wrapfig}

\title{Two-Page Book Embeddings of 4-Planar Graphs\thanks{Work on this problem began at Dagstuhl Seminar 13151. We thank the organizers, participants and Prof. Dr. M. Kaufmann.}}

\author{Michael A. Bekos%
\thanks{Electronic address: \texttt{bekos@informatik.uni-tuebingen.de}}}
\affil{Wilhelm-Schickhard-Institut f\"ur Informatik, Universit\"at T\"ubingen, Germany}

\author{Martin~Gronemann%
\thanks{Electronic address: \texttt{gronemann@informatik.uni-koeln.de}}}
\affil{Institut f\"ur Informatik, Universit\"at zu K\"oln, Germany}

\author{Chrysanthi~N.~Raftopoulou%
\thanks{Electronic address: \texttt{crisraft@mail.ntua.gr}}}
\affil{School of Applied Mathematical \& Physical Sciences, NTUA, Greece.}

\date{}

\makeatletter
\def\@maketitle{%
  \newpage
  \null
  \vskip 2em%
  \begin{center}%
  \let \footnote \thanks
    {\Large\bfseries \@title \par}%
    \vskip 1.5em%
    {\normalsize
      \lineskip .5em%
      \begin{tabular}[t]{c}%
        \@author
      \end{tabular}\par}%
    \vskip 1em%
    {\normalsize \@date}%
  \end{center}%
  \par
  \vskip 1.5em}
\makeatother

\newtheorem{lemma}{Lemma}
\newtheorem{theorem}{Theorem}
\newtheorem{corollary}{Corollary}

\newcommand{\inGraph}[1]{G_{in}(#1)}
\newcommand{\outGraph}[1]{G_{out}(#1)}
\newcommand{\inGraphDummy}[1]{G^*_{in}(#1)}
\newcommand{\outGraphDummy}[1]{G^*_{out}(#1)}

\newcommand{\inGraphInc}[1]{\overline{G}_{in}(#1)}
\newcommand{\outGraphInc}[1]{\overline{G}_{out}(#1)}
\newcommand{\inHamilton}[1]{{H}_{in}(#1)}
\newcommand{\outHamilton}[1]{{H}_{out}(#1)}
\newcommand{\inHamiltonDummy}[1]{{H}^*_{in}(#1)}
\newcommand{\outHamiltonDummy}[1]{{H}^*_{out}(#1)}

\newcommand{\WLOG}{{w.l.o.g. }}

\newcommand{\an}{anchor\xspace}
\newcommand{\cc}{ancillary\xspace}
\newcommand{\Cc}{Ancillary\xspace}
\newcommand{\bv}{block-vertex\xspace}
\newcommand{\ans}{anchors\xspace}
\newcommand{\ccs}{ancillaries\xspace}
\newcommand{\bvs}{block-vertices\xspace}

\begin{document}

\maketitle

\begin{abstract}
Back in the eighties, Heath~\cite{heath-thesis-85} showed that every
3-planar graph is subhamiltonian and asked whether this result can
be extended to a class of graphs of degree greater than three. In
this paper we affirmatively answer this question for the class of
4-planar graphs. Our contribution consists of two algorithms: The
first one is limited to triconnected graphs, but runs in linear time
and uses existing methods for computing hamiltonian cycles in planar
graphs. The second one, which solves the general case of the
problem, is a quadratic-time algorithm based on the book embedding
viewpoint of the problem.
\end{abstract}
\section{Introduction}
\label{sec:intro}
Book embeddings have a long history and arise in various application
areas such as VLSI design, parallel computing, design of
fault-tolerant systems~\cite{Chung87embeddinggraphs}. In a
\emph{book embedding} the placement of nodes is restricted to a
line, the \emph{spine} of the book. The edges are assigned to
different \emph{pages} of the book. A page can be thought of as a
half-plane bounded by the spine where the edges are drawn as
circular arcs between their endpoints. We say that a graph admits a
\emph{$k$-page book embedding} or is \emph{k-page embeddable} if one
can assign the edges to $k$ pages and there exists a linear ordering
of the nodes on the spine such that no two edges of the same page
cross. The minimum number of pages required to construct such an
embedding is the \emph{book thickness} or \emph{page number} of a
graph. The book thickness of planar graphs has received much
attention in the past. Yannakakis~\cite{JCSS::Yannakakis1989}
describes a linear-time algorithm to embed every planar graph into a
book of four pages. We study the problem of embedding 4-planar
graphs, i.e., planar graphs with maximum degree four, into books
with two pages. Bernhart et al.~\cite{bk-btg-79} show that a graph
is two-page embeddable iff it is subhamiltonian. A
\emph{subhamiltonian graph} is a subgraph of a planar hamiltonian
graph. It is \textit{NP}-complete to determine whether a graph is
subhamiltonian~\cite{wigderson82}. Often referred to as
\emph{augmented hamiltonian cycle}, a \emph{subhamiltonian cycle} is
a cyclic sequence of nodes in a graph that would form a hamiltonian
cycle when adding the missing edges without destroying planarity.
The relation between subhamiltonian cycles and two-page book
embeddings is quite intuitive. The order of the nodes on the spine
is equivalent to the cyclic order of the subhamiltonian cycle. The
edges are partitioned by whether they lie in the interior of the
cycle or not.

An early important result is due to Whitney~\cite{Whi}, who proves
that every maximal planar graph with no separating triangles is
hamiltonian (recall that a \emph{separating triangle} is a 3-cycle
whose removal disconnects the graph). Tutte~\cite{Tut56} shows that
every 4-connected planar graph has a hamiltonian cycle.
Chiba~et~al.~\cite{chibanishizeki89} provide a linear-time algorithm
to find a hamiltonian cycle in a 4-connected planar graph.
Chen~\cite{Chen:2003fk} gives a proof that every maximal planar
graph with at least five vertices and no separating triangles is
4-connected. Sanders~\cite{sanders-97} generalizes a theorem of
Thomassen and shows that any 4-connected planar graph has a
hamiltonian cycle that contains two arbitrarily chosen edges of the
graph. Based on Whitney's theorem,
Kainen~et~al.~\cite{Kainen2007835} show that every planar graph with
no separating triangles is subhamiltonian. Another result is by
Chen~\cite{Chen:2003fk} who shows that if a maximal planar graph
contains only one such triangle, then it is hamiltonian.
Helden~\cite{Helden20071833} improves this result further to two
triangles. The aforementioned results are all related to the problem
of embedding planar graphs into two pages. However, there is an
extensive amount of literature on embedding various types of graphs
into books; for an overview see e.g.~\cite{DujmovicW04}. One result
that is interesting in our context is that of
Heath~\cite{heath-thesis-85}. In his thesis, he describes a
linear-time algorithm to embed any 3-planar graph into two pages and
concludes that it would be interesting to know if a higher degree
bound is possible.

We tackle the 4-planar case from two sides. The first approach based
on the subhamiltonicity is restricted to triconnected graphs
(Section~\ref{sec:triconnected-planar}) but builds on existent
results and is therefore of a simple nature compared to the second
approach. Extending it to biconnected graphs is not straightforward,
though. The algorithm of Section~\ref{sec:general-planar} --which is
less efficient in terms of time complexity-- exploits the degree
restriction to construct a two-page book embedding.

\section{Subhamiltonicity of Triconnected 4-Planar Graphs}
\label{sec:triconnected-planar}
In this section we restrict ourselves to triconnected 4-planar
graphs. To state the main result of this section, we proceed in a
step-by-step manner. First we investigate the special properties of
separating triangles in 4-planar graphs, then we use those to derive
a solution for a single separating triangle. Unlike
Chen~\cite{Chen:2003fk} and Helden~\cite{Helden20071833}, we are
able to extend our approach to an unbounded number of triangles by
exploiting the degree restriction. We say a subhamiltonian cycle $H$
\emph{crosses} a face if there are two consecutive vertices in $H$
that are incident to the face but not adjacent to each other.

\begin{figure}[t]
    \centering
    \begin{minipage}[b]{0.4\textwidth}
        \centering
        \includegraphics[width=0.65\textwidth,page=1]{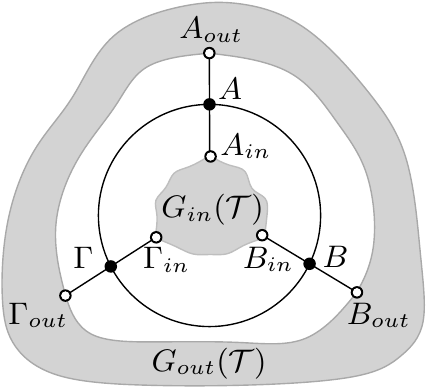}
        \caption{Triangle $\mathcal{T}$ separating $\inGraph{\mathcal{T}}$ and $\outGraph{\mathcal{T}}$ on removal.}
        \label{fig:sep_tri_4planar}
    \end{minipage}
    \hfill
    \begin{minipage}[b]{0.5\textwidth}
        \centering
        \includegraphics[width=0.65\textwidth,page=2]{images/sep_tri}
        \caption{Merging $\inHamilton{\mathcal{T}}$ (dotted) and $\outHamilton{\mathcal{T}}$ (dashed) into $H$ (bold gray).}
        \label{fig:merging_cycles}
    \end{minipage}
\end{figure}

\begin{lemma}
\label{lemma:linear_time_one_crossing} Every triconnected planar
graph with no separating triangles has a subhamiltonian cycle that
crosses every face at most once and it can be computed in linear
time.
\end{lemma}
\begin{proof}
In the triconnected case, Kainen~et~al.~\cite{Kainen2007835}
construct a new maximal planar graph $G'=(V',E')$ by inserting a
vertex into each non-triangular face of $G$ and connect it to the
vertices of that face. Clearly this takes linear time. $G'$ is
maximal planar, free of separating triangles, hence, 4-connected. We
can use the linear-time algorithm of
Chiba~et~al.~\cite{chibanishizeki89} to obtain a hamiltonian cycle
$H'$ for $G'$. Deleting the newly inserted vertices $V' - V$ yields
a subhamiltonian cycle $H$ for $G$ that crosses each face at most
once.
\end{proof}

Before investigating the properties of separating triangles, we
introduce some notation. Given an embedded triconnected 4-planar
graph $G$ with a fixed outerface and a separating triangle
$\mathcal{T}$ with vertices $V(\mathcal{T})= \{A, B, \Gamma\} $, we
denote the subgraph of $G$ contained in $\mathcal{T}$ by
$\inGraph{\mathcal{T}}$ and the subgraph of $G$ outside
$\mathcal{T}$ by $\outGraph{\mathcal{T}}$. We also denote
$\inGraphInc{\mathcal{T}} = G - \outGraph{\mathcal{T}}$ and
$\outGraphInc{\mathcal{T}} = G - \inGraph{\mathcal{T}}$. Since $G$
is triconnected and 4-planar, every vertex of $\mathcal{T}$ has
degree four and is adjacent to exactly one vertex in
$\inGraph{\mathcal{T}}$ and $\outGraph{\mathcal{T}}$, respectively.
We denote these with $A_{in}, B_{in}, \Gamma_{in}$ and $A_{out},
B_{out}, \Gamma_{out}$, respectively (see
Fig.~\ref{fig:sep_tri_4planar}).

\begin{lemma}\label{lemma:abc_distinct}
Given a 4-planar triconnected graph $G$ and a separating triangle
$\mathcal{T} = \{A,B,\Gamma\}$, then $A_{in}, B_{in}, \Gamma_{in}
(A_{out}, B_{out}, \Gamma_{out})$ are pairwise distinct or all
represent the same vertex.
\end{lemma}
\begin{proof}
In the other case, where \WLOG $A_{in} = B_{in} = v$ and
$\Gamma_{in} \neq v$, there exists a separation pair $(v,\Gamma)$
contradicting the triconnectivity of $G$. A symmetric argument
applies to $A_{out}, B_{out}, \Gamma_{out}$.
\end{proof}

\begin{lemma}\label{lemma:triangle_disjoint}
In a 4-planar triconnected graph, every pair of distinct separating
triangles $\mathcal{T}$ and $\mathcal{T}'$ is vertex disjoint, i.e.
$V(\mathcal{T}) \cap V(\mathcal{T}') = \emptyset$.
\end{lemma}
\begin{proof}
Assume to the contrary that $\mathcal{T}$ and $\mathcal{T}'$ share
an edge or a vertex. In the first case, let \WLOG $e = (u,v)$ be the
common edge. The degree of both $u$ and $v$ is at least five, since
three edges are required for $\mathcal{T}, \mathcal{T}'$ and two
additional edges to connect  $\inGraph{\mathcal{T}}$ and
$\inGraph{\mathcal{T}'}$ to $\mathcal{T}$ and $\mathcal{T}'$,
respectively. In the second case, let $v$ denote the common vertex.
Since $v$ is part of two edge disjoint cycles and connected to
$\inGraph{\mathcal{T}}$ and $\inGraph{\mathcal{T}'}$, it follows
that $deg(v) \geq 6$.
\end{proof}

Consider now a 4-planar triconnected graph with a single separating
triangle $\mathcal{T}$. Similar to Chen~\cite{Chen:2003fk}, the idea
is to compute two cycles $\inHamilton{\mathcal{T}}$ and
$\outHamilton{\mathcal{T}} $ for $\inGraphInc{\mathcal{T}}$ and
$\outGraphInc{\mathcal{T}}$ and link them via the separating
triangle together. The crucial observation is that if two cycles
intersect as illustrated in Fig.~\ref{fig:merging_cycles}, i.e.,
they contain two edges of the triangle but have only one of them in
common, then we can always merge them into one cycle.
\begin{lemma}
\label{lemma:cycle_merging} Let $G$ be a triconnected 4-planar
graph, $\mathcal{T}$ a separating triangle, and
$\inHamilton{\mathcal{T}}$ and $\outHamilton{\mathcal{T}} $ two
subhamiltonian cycles for $\inGraphInc{\mathcal{T}}$ and
$\outGraphInc{\mathcal{T}}$, resp. If $E(\inHamilton{\mathcal{T}})
\cap E(\mathcal{T})=\{e_{in},e\}$ and $E(\outHamilton{\mathcal{T}})
\cap  E(\mathcal{T})=\{e_{out},e\} $ where $\{e, e_{in} ,e_{out}\}$
are the edges of $\mathcal{T}$, then $G$ is subhamiltonian.
\end{lemma}
\begin{proof}
Let \WLOG $e = (A,B)$, $e_{in}=(B,\Gamma)$ and $e_{out}=(A,\Gamma) $
as illustrated in Fig.~\ref{fig:merging_cycles}. The result of
removing the edges of $\mathcal{T}$ from both cycles are two paths
$P_{out} = B \rightsquigarrow \Gamma$ and $P_{in} = \Gamma
\rightsquigarrow A$. Joining them at $\Gamma$ and inserting $e$
yields a subhamiltonian cycle.
\end{proof}

\begin{figure}[t]
\centering
    \begin{minipage}[b]{.20\textwidth}
        \centering
        \subfloat[\label{fig:sep_tri_dummy}{$v_\mathcal{T}$ in $\outGraphDummy{\mathcal{T}}$}]
        {\includegraphics[width=\linewidth,page=3]{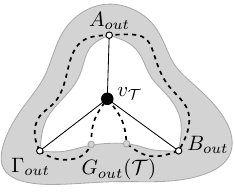}}
    \end{minipage}
    \hfill
      \begin{minipage}[b]{.20\textwidth}
        \centering
        \subfloat[\label{fig:sep_tri_dummy_T}{$\mathcal{T}$ in $\outGraphInc{\mathcal{T}}$}]
        {\includegraphics[width=\linewidth,page=4]{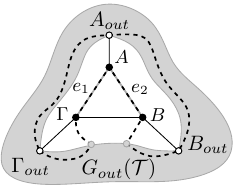}}
       \end{minipage}
       \hfill
    \begin{minipage}[b]{.15\textwidth}
        \centering
        \subfloat[\label{fig:sep_tri_inside_cycle}{$v'_\mathcal{T}$ in $\inGraphDummy{\mathcal{T}}$}]
        {\includegraphics[width=\linewidth,page=5]{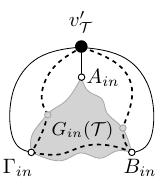}}
    \end{minipage}
    \hfill
    \begin{minipage}[b]{.20\textwidth}
        \centering
        \subfloat[\label{fig:sep_tri_result_cycle}{$\mathcal{T}$ in $\inGraphInc{\mathcal{T}}$}]
        {\includegraphics[width=\linewidth,page=6]{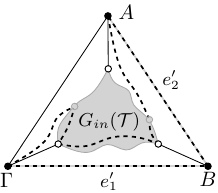}}
    \end{minipage}
    \begin{minipage}[b]{.20\textwidth}
        \centering
        \subfloat[\label{fig:sep_tri_result_cycle}{$G$ with $\mathcal{T}$ and $H$}]
        {\includegraphics[width=\linewidth,page=7]{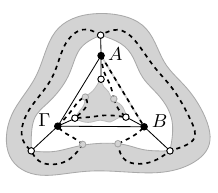}}
    \end{minipage}
    \caption{
    (a)~Subhamiltonian cycle $\outHamiltonDummy{\mathcal{T}}$ in $\outGraphDummy{\mathcal{T}}$ containing $v_\mathcal{T}$.
    (b)~Augmenting $\outHamiltonDummy{\mathcal{T}}$ yields $\outHamilton{\mathcal{T}}$ containing edges $e_1 = (\Gamma,A)$ and $e_2 = (A,B)$.
    (c)~Dummy vertex $v'_\mathcal{T}$ as replacement for $\mathcal{T}$ in $\inGraphDummy{\mathcal{T}}$ and a cycle $\inHamiltonDummy{\mathcal{T}}$.
    (d)~Rerouting $\inHamiltonDummy{\mathcal{T}}$ through $\mathcal{T}$ resulting in $\inHamilton{\mathcal{T}}$ with edges $e'_1 = (\Gamma,B)$ and $e_2 = (A,B)$.
    (d)~The result of merging $\inHamilton{\mathcal{T}}$ and $\outHamilton{\mathcal{T}}$ into a cycle $H$ for $G$.}
\end{figure}

It remains to show that we can always find two cycles that satisfy
the requirements of Lemma~\ref{lemma:cycle_merging}. In the
following, we neglect the degenerated case of
Lemma~\ref{lemma:abc_distinct}, where $\outGraph{\mathcal{T}}$ or
$\inGraph{\mathcal{T}}$ is a single vertex, because finding a cycle
in that case is trivial. Consider for example
$\outGraphInc{\mathcal{T}}$, for $\inGraphInc{\mathcal{T}}$ a
symmetric argument holds. To obtain $\outHamilton{\mathcal{T}}$, we
temporarily replace $\mathcal{T}$ in $\outGraphInc{\mathcal{T}}$
with a single vertex $v_\mathcal{T}$ as depicted in
Fig.~\ref{fig:sep_tri_dummy}. The resulting graph
$\outGraphDummy{\mathcal{T}}$ remains 4-planar and triconnected,
because $deg(v_\mathcal{T}) = 3$ by construction and any path via
$\mathcal{T}$ can use $v_\mathcal{T}$ instead. One may argue that
this operation may introduce additional separating triangles.
However, such a triangle must contain $v_\mathcal{T}$ and,
therefore, $deg(v_\mathcal{T}) = 4$, a contradiction. Now let us
assume that $\outHamiltonDummy{\mathcal{T}}$ is a subhamiltonian
cycle for $\outGraphDummy{\mathcal{T}}$. The idea is to reinsert
$\mathcal{T}$ and reroute $\outHamiltonDummy{\mathcal{T}}$ through
$\mathcal{T}$ such that the resulting cycle
$\outHamilton{\mathcal{T}}$ contains two edges $e_1,e_2 \in
E(\mathcal{T})$.

\begin{lemma}
\label{lemma:cycle_rerouting} Let $G$ be a triconnected 4-planar
graph, $\mathcal{T}$ a separating triangle. Furthermore, let
$\outGraphDummy{\mathcal{T}}$ denote the graph resulting from
replacing $\mathcal{T}$ by a vertex $v_\mathcal{T}$ in
$\outGraphInc{\mathcal{T}}$. A subhamiltonian cycle
$\outHamiltonDummy{\mathcal{T}}$ for $\outGraphDummy{\mathcal{T}}$
can be augmented to a subhamiltonian cycle
$\outHamilton{\mathcal{T}}$ for $\outGraphInc{\mathcal{T}}$ such
that it contains two edges of $\mathcal{T}$, i.e.,
$E(\outHamilton{\mathcal{T}}) \cap E(\mathcal{T}) = \{e_1, e_2\}$.
If $\outHamiltonDummy{\mathcal{T}}$ crosses every face of
$\outGraphDummy{\mathcal{T}}$ at most once, one may choose any pair
$e_1, e_2 \in E(\mathcal{T})$ to lie on $\outHamilton{\mathcal{T}}$.
\end{lemma}
\begin{proof}
To prove the claim, it is sufficient to consider every combination
of $e_1, e_2$ and the location of the predecessor and successor of
$v_\mathcal{T}$ in $\outHamiltonDummy{\mathcal{T}}$. 
In the following, we enumerate and describe in detail all possible
cases that occur when augmenting $\outHamiltonDummy{\mathcal{T}}$
such that the resulting cycle $\outHamilton{\mathcal{T}}$ contains
two edges $e_1, e_2$ of $\mathcal{T}$. 
To avoid any redundancies, we omit
symmetric cases and consider for the same reason a directed cycle.
We distinguish between three main cases depending on the location of
the predecessor and successor of $v_\mathcal{T}$ in
$\outHamiltonDummy{\mathcal{T}}$.

\begin{figure}[h!]
\centering
    \begin{minipage}[b]{.32\textwidth}
        \centering
        \subfloat[\label{fig:dummy_edge_edge}{Case 1}]
        {\includegraphics[page=1]{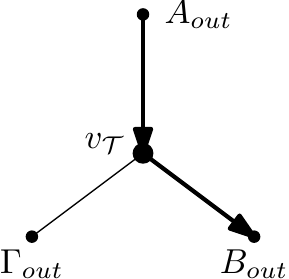}}
    \end{minipage}\hfill
      \begin{minipage}[b]{.32\textwidth}
        \centering
        \subfloat[\label{fig:dummy_edge_face}{Case 2}]
        {\includegraphics[page=2]{images/sep_tri_reroute_cases}}
       \end{minipage} \hfill
    \begin{minipage}[b]{.32\textwidth}
        \centering
        \subfloat[\label{fig:dummy_face_face}{Case 3}]
        {\includegraphics[page=3]{images/sep_tri_reroute_cases}}
       \end{minipage}\hfill
        \caption{The three main cases at $v_\mathcal{T}$: (a)~The cycle uses two of the three edges incident to $v_\mathcal{T}$.
        (b)~The cycle enters via an edge and leaves through a face.
        (c)~Predecessor and successor are not adjacent to $v_\mathcal{T}$.}
        \label{fig:sep_tri_reroute_cases}
\end{figure}

\begin{figure}[t]
\centering
    \begin{minipage}[b]{.32\textwidth}
        \centering
        \subfloat[\label{fig:reroute_edge_edge_ACB}{$(A,\Gamma)$ and $(B,\Gamma)$}]
        {\includegraphics[page=4]{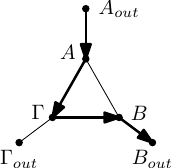}}
    \end{minipage}
    \hfill
    \begin{minipage}[b]{.32\textwidth}
        \centering
        \subfloat[\label{fig:reroute_edge_edge_CAB}{$(A,B)$ and $(A,\Gamma)$}]
        {\includegraphics[page=5]{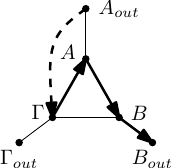}}
       \end{minipage} \hfill
    \begin{minipage}[b]{.32\textwidth}
        \centering
        \subfloat[\label{fig:reroute_edge_edge_ABC}{$(A,B)$ and $(B,\Gamma)$}]
        {\includegraphics[page=6]{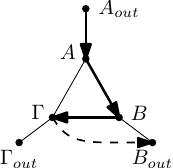}}
    \end{minipage}
    \hfill
    \caption{
    (a)~Dummy vertex $v_T$ is replaced by the sequence $A,\Gamma,B$ to obtain a cycle with the edges $(A,\Gamma)$ and $(B,\Gamma)$.
    (b)~Sequence $\Gamma,A,B$ yields a cycle with $(A,B)$ and $(A,\Gamma)$ where $(A_{out},\Gamma)$ requires it to cross a face.
    (c)~Augmenting with $A,B,\Gamma$ results in a cycle containing $(A,B)$ and $(B,\Gamma)$.}
    \label{fig:sep_tri_reroute_case_1}
    \begin{minipage}[b]{.32\textwidth}
        \centering
        \subfloat[\label{fig:reroute_edge_face_ACB}{$(A,\Gamma)$ and $(B,\Gamma)$}]
        {\includegraphics[page=7]{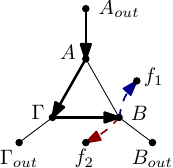}}
    \end{minipage}
    \hfill
    \begin{minipage}[b]{.32\textwidth}
        \centering
        \subfloat[\label{fig:reroute_edge_face_CAB}{$(A,B)$ and $(A,\Gamma)$}]
        {\includegraphics[page=8]{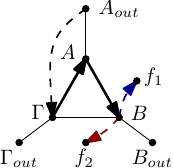}}
    \end{minipage}
    \hfill
    \begin{minipage}[b]{.32\textwidth}
        \centering
        \subfloat[\label{fig:reroute_edge_face_ABC}{$(A,B)$ and $(B,\Gamma)$}]
        {\includegraphics[page=9]{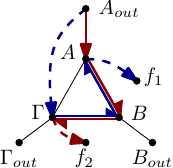}}
    \end{minipage}
    \hfill
    \caption{
    In (a)~and~(b) the same sequences as before are used to obtain a cycle containing $(A,\Gamma), (B,\Gamma)$ and $(A,B), (A,\Gamma)$, respectively.
    Subcase-specific links are drawn in blue and red. (c) A more complicated case requiring one additional crossing of a face from $A_{out}$ to $\Gamma$.}
    \label{fig:sep_tri_reroute_case_2}
     \begin{minipage}[b]{.32\textwidth}
        \centering
        \subfloat[\label{fig:reroute_face_face_ACB}{$(A,\Gamma)$ and $(B,\Gamma)$}]
        {\includegraphics[page=10]{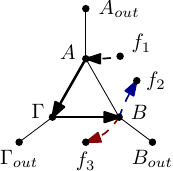}}
    \end{minipage}\hfill
    \begin{minipage}[b]{.32\textwidth}
        \centering
        \subfloat[\label{fig:reroute_face_face_CAB}{$(A,B)$ and $(A,\Gamma)$}]
        {\includegraphics[page=11]{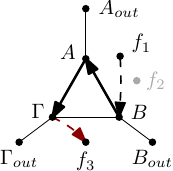}}
    \end{minipage} \hfill
    \begin{minipage}[b]{.32\textwidth}
        \centering
        \subfloat[\label{fig:reroute_face_face_ABC}{$(A,B)$ and $(B,\Gamma)$}]
        {\includegraphics[page=12]{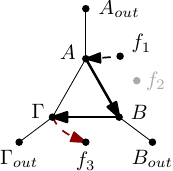}}
    \end{minipage}\hfill
    \caption{(a)~Both subcases have a solution. (b,c) When the cycle uses two distinct faces ($f_1\rightsquigarrow f_3$) a solution for both pairs of edges can be found. If only one face is used ($f_1\rightsquigarrow f_2$), then no solution exists for the edges $(A,B), (A,\Gamma)$ and $(A,B), (B,\Gamma)$.}
    \label{fig:sep_tri_reroute_case_3}
\end{figure}

\begin{description}
\item[Case 1] ($\text{Edge} \rightsquigarrow v_\mathcal{T} \rightsquigarrow
\text{Edge}$): Both the predecessor and successor of $v_\mathcal{T}$
in $\outHamiltonDummy{\mathcal{T}}$ are adjacent to $v_\mathcal{T}$,
hence, the cycle $\outHamiltonDummy{\mathcal{T}}$ contains two edges
incident to $v_\mathcal{T}$, let us say $(A_{out}, v_\mathcal{T}),
(v_\mathcal{T}, B_{out})$ as illustrated in
Fig.~\ref{fig:dummy_edge_edge}.
Fig.~\ref{fig:sep_tri_reroute_case_1} depicts how
$\outHamiltonDummy{\mathcal{T}}$ can be augmented such that every
pair of edges of $T$ is contained in $\outHamilton{\mathcal{T}}$.
Notice that while for the pair $(A,\Gamma),(B,\Gamma)$ in
Fig~\ref{fig:reroute_edge_edge_ACB} no face crossing is required,
for the two other pairs one additional face crossing is introduced
(Fig.~\ref{fig:reroute_edge_edge_CAB} and
\ref{fig:reroute_edge_edge_ABC}).

\item[Case 2] ($\text{Edge} \rightsquigarrow v_\mathcal{T} \rightsquigarrow
\text{Face}$): In this case, the predecessor, say $A_{out}$, is
adjacent to $v_\mathcal{T}$, while the successor is not. Since
$\outHamiltonDummy{\mathcal{T}}$ is a subhamiltonian cycle, the
successor is incident to one of the three faces incident to
$v_\mathcal{T}$. To cover all possible combinations, we distinguish
between whether (i) the predecessor $A_{out}$ is incident to that
face or (ii) not. Fig.~\ref{fig:dummy_edge_face} illustrates both
configurations, where $f_1$ denotes the successor located at a face
of type (i), and $f_2$ the successor that is incident to the face at
the opposite side (ii). For both subcases, the rerouting rules for
the first two edge pairs are relatively simple, since they follow
the basic principle of the first case, see
Fig.~\ref{fig:reroute_edge_face_ACB}
and~\ref{fig:reroute_edge_face_CAB}. However, the third pair is more
complicated. For (i) the sequence $A_{out}, v_\mathcal{T} , f_1$ is
replaced by $A_{out}, \Gamma, B, A, f_1$, whereas for (ii) $A_{out},
v_\mathcal{T} , f_1$ is substituted by $A_{out}, A, B, \Gamma, f_2$
(Fig.~\ref{fig:reroute_edge_face_ABC}).

\item[Case 3] $\text{Face} \rightsquigarrow v_\mathcal{T} \rightsquigarrow
\text{Face}$. Both predecessor and successor of $v_\mathcal{T}$ in
$\outHamiltonDummy{\mathcal{T}}$ are not adjacent. Hence, the cycle
enters and leaves $v_T$ through a face. Again to cover all
possibilities, we have to deal with two subcases: (i) the two faces
are distinct or (ii) the cycle $\outHamiltonDummy{\mathcal{T}}$
leaves through the same face as it enters. Rerouting
$\outHamiltonDummy{\mathcal{T}}$ in the first subcase (i) works for
all three different edge pairs, even without introducing any new
face crossings. The three solutions for (i) are displayed in
Fig.~\ref{fig:sep_tri_reroute_case_3}, where the predecessor is
labeled by $f_1$ and the successor by $f_3$. So far we have been
able to resolve every configuration such that any pair of edges can
be selected to be part of $\outHamilton{\mathcal{T}}$. However, the
interesting case is subcase (ii), where the predecessor $f_1$ and
successor $f_2$ are incident to the same face. While there is a
solution for the edge pair $(A,\Gamma), (B,\Gamma)$ as displayed in
Fig.~\ref{fig:reroute_face_face_ACB}, the two remaining edge pairs
create unresolvable configurations, see
Fig.~\ref{fig:reroute_face_face_CAB}
and~\ref{fig:reroute_face_face_ABC}, respectively. This dilemma is
caused by the fact that $\outHamilton{\mathcal{T}}$ has to either
enter or leave $\mathcal{T}$ via $\Gamma$. However, $\Gamma$ is not
accessible from neither $f_1$ nor $f_2$ without destroying
planarity.
\end{description}

We may summarize the solutions for the different cases as follows:
As long as the cycle does not enter and leave $v_\mathcal{T}$ via
the same face, we can always choose two edges of $\mathcal{T}$ in
advance and reroute the cycle such that these two edges become part
of $H$. 
\end{proof}
At this point it is tempting to show that we can always find
a cycle that avoids crossing a face twice. By using
Lemma~\ref{lemma:linear_time_one_crossing}, we may obtain such a
cycle in a triconnected graph with no separating triangles. This
raises the question if we can use it and apply the described rules
to obtain a cycle through multiple triangles for which we may
specify two edges in advance. We answer this question negatively
with a small counterexample.

\begin{figure}[t]
    \centering
    \begin{minipage}[b]{.32\textwidth}
        \centering
        \subfloat[\label{fig:sep_tri_reroute_counterexample_0}{$G$ with $\mathcal{T}$ and $\mathcal{T}'$}]
        {\includegraphics[page=1]{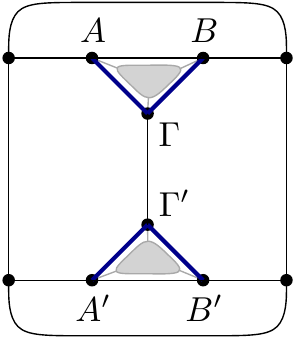}}
    \end{minipage}\hfill
      \begin{minipage}[b]{.32\textwidth}
        \centering
        \subfloat[\label{fig:sep_tri_reroute_counterexample_1}{Cycle from Lemma~\ref{lemma:linear_time_one_crossing}}]
        {\includegraphics[page=2]{images/sep_tri_reroute_counterexample}}
       \end{minipage} \hfill
    \begin{minipage}[b]{.32\textwidth}
        \centering
        \subfloat[\label{fig:sep_tri_reroute_counterexample_2}{After rerouting at $\mathcal{T}$ }]
        {\includegraphics[page=3]{images/sep_tri_reroute_counterexample}}
    \end{minipage}\hfill
    \caption{
    (a)~Two separating triangles $\mathcal{T}$ and $\mathcal{T}'$ with vertices $V(\mathcal{T}) = \{A,B,\Gamma\}$ and $V(\mathcal{T}') = \{A',B',\Gamma'\}$ and for each two prescribed edges (bold, blue).
    (b)~$\mathcal{T}$ and $\mathcal{T}'$ replaced by $v_\mathcal{\mathcal{T}}$ and $v_{\mathcal{T}'}$, every non triangular face is stellated by inserting additional vertices (squares) and
    edges (dashed), and a (sub)hamiltonian cycle $H'$ (bold, red).
    (c)~Result of applying the corresponding rule to $\mathcal{T}$ creating an unresolvable configuration for $\mathcal{T}'$.}
    \label{fig:sep_tri_reroute_counterexample}
\end{figure}

Consider the triconnected 4-planar graph $G$ shown in
Fig.~\ref{fig:sep_tri_reroute_counterexample_0}. It contains two
separating triangles $\mathcal{T}$ and $\mathcal{T}'$ with vertices
$V(\mathcal{T}) = \{A,B,\Gamma\}$ and $V(\mathcal{T}') =
\{A',B',\Gamma'\}$, respectively. In every triangle two edges (bold,
blue) are prescribed to lie on the augmented subhamiltonian cycle
$H$. We proceed as described; both triangles are replaced by a dummy
vertex $v_{\mathcal{T}}$ and $v_{\mathcal{T}'}$, respectively. The
resulting graph (Fig.~\ref{fig:sep_tri_reroute_counterexample_1}) is
triconnected 4-planar and free of separating triangles. The squares
and dashed lines  correspond to the dummy vertices and edges
inserted by the technique of Kainen et al.~\cite{Kainen2007835} as
described in Lemma~\ref{lemma:linear_time_one_crossing}. We may now
compute a hamiltonian cycle $H'$ by applying the linear-time
algorithm of Chiba et al.~\cite{chibanishizeki89}. Assume the result
is the bold cycle in
Fig.~\ref{fig:sep_tri_reroute_counterexample_1}. Clearly the cycle
crosses every face at most once after we remove the dummy vertices
inserted by the technique of Kainen et al.~\cite{Kainen2007835}. We
reinsert $\mathcal{T}$ and apply the corresponding rule, i.e., the
augmentation displayed in Fig.~\ref{fig:reroute_edge_face_CAB}. The
result of augmenting such that the two marked edges of
$\mathcal{T}$, namely $(A,\Gamma), (B,\Gamma)$, lie on the cycle is
displayed in Fig.~\ref{fig:sep_tri_reroute_counterexample_2}. Notice
that we are forced to enter $\mathcal{T}$ via $A$ and exit by $B$.
As a result, the cycle crosses one face twice. Moreover,
$\mathcal{T}'$ must be entered and left through the same face. The
corresponding rule, illustrated in
Fig.~\ref{fig:reroute_face_face_CAB}, implies that we cannot reroute
the cycle such that it contains the edges $(A',\Gamma'),
(B',\Gamma')$. However, we may lift the restriction, use the only
rule applicable in this case (Fig.~\ref{fig:reroute_face_face_ACB}),
and obtain a cycle with edges $(A',\Gamma'), (A',B')$ instead.
Notice that the graph in this example has even a hamiltonian cycle $H$ through
the requested edges. However, the purpose of the example is to
demonstrate that for an arbitrary chosen subhamiltonian cycle, the
described rules cannot always be applied. We may conclude that when
using Lemma~\ref{lemma:linear_time_one_crossing}, we may choose for
one (the first) triangle two edges because the initial cycle visits
every face at most once. From there on, we can only guarantee that
two unknown edges are part of the final cycle. In the following we will benefit from this observation.

Recall the aforementioned single-separating-triangle scenario. Both
$\outGraphInc{\mathcal{T}}$ and $\inGraphInc{\mathcal{T}}$ are free
of separating triangles. Therefore, we may construct two graphs
$\outGraphDummy{\mathcal{T}}, \inGraphDummy{\mathcal{T}}$ by
replacing $\mathcal{T}$ with dummy vertices. Applying
Lemma~\ref{lemma:linear_time_one_crossing} to them yields two
subhamiltonian cycles $\outHamiltonDummy{\mathcal{T}}$ and
$\inHamiltonDummy{\mathcal{T}}$, both crossing every face of
$\outGraphDummy{\mathcal{T}}$ and $\inGraphDummy{\mathcal{T}}$ at
most once. Hence, we may augment them with the aid of
Lemma~\ref{lemma:cycle_rerouting} such that they contain each two
edges of $\mathcal{T}$. By choosing the combination of the edges such
that $\outHamilton{\mathcal{T}}$ and $\inHamilton{\mathcal{T}}$ meet
the requirements of Lemma~\ref{lemma:cycle_merging}, we can
merge them into a single subhamiltonian cycle $H$ for $G$.

While the property that $\outGraphDummy{\mathcal{T}}$ and
$\inGraphDummy{\mathcal{T}}$ are both free of separating triangles
enables us to conveniently choose two edges for each cycle
$\outHamilton{\mathcal{T}}, \inHamilton{\mathcal{T}}$, this only
works for a single separating triangle. However, a closer look
reveals that it is sufficient to have a choice for either
$\outHamilton{\mathcal{T}}$ or  $\inHamilton{\mathcal{T}}$, not
necessarily both of them. The idea is to first augment the cycle for
which we do not have a choice to see which edges of $\mathcal{T}$
are part of it, then we choose the edges for the second cycle
accordingly. We summarize the idea as the main result of this
section and describe it in a more formal manner in form of a proof.

\begin{theorem}\label{theorem:triconnected}
Every triconnected 4-planar graph is subhamiltonian.
\end{theorem}
\begin{proof}
Let $G$ denote a triconnected 4-planar graph and $\tau(G)$ the
number of separating triangles in $G$. We prove by induction and
claim that for any $\tau(G) \geq 0$ we can compute a subhamiltonian
cycle $H$ for $G$. \emph{Base case}: Since $\tau(G) = 0$, we can
directly apply Lemma~\ref{lemma:linear_time_one_crossing}.
\emph{Inductive case:} For $\tau(G) > 0$, we pick a separating
triangle $\mathcal{T}$ such that $\tau(\inGraphInc{\mathcal{T}}) =
0$. Let $\outGraphDummy{\mathcal{T}}$ be the result of replacing
$\mathcal{T}$ by $v_\mathcal{T}$ in $\outGraphInc{\mathcal{T}}$.
Notice that $\tau(\outGraphDummy{\mathcal{T}}) = \tau(G) - 1$ holds.
Hence, by induction hypothesis, $\outGraphDummy{\mathcal{T}}$ has a
subhamiltonian cycle $\outHamiltonDummy{\mathcal{T}}$. We reinsert
$\mathcal{T}$ and augment $\outHamiltonDummy{\mathcal{T}}$ such that
the result $\outHamilton{\mathcal{T}}$ contains two (arbitrary)
edges $e_1, e_2$ of $\mathcal{T}$. In a similar way, we replace
$\mathcal{T}$ in $\inGraphInc{\mathcal{T}}$ by $v'_\mathcal{T}$ to
obtain $\inGraphDummy{\mathcal{T}}$. Since
$\tau(\inGraphInc{\mathcal{T}}) = \tau(\inGraphDummy{\mathcal{T}}) =
0$ holds, we can apply Lemma~\ref{lemma:linear_time_one_crossing} to
$\inGraphDummy{\mathcal{T}}$ and compute a cycle
$\inHamiltonDummy{\mathcal{T}}$ that crosses each face at most once.
With Lemma~\ref{lemma:cycle_rerouting} we may obtain a cycle
$\inHamilton{\mathcal{T}}$ for $\inGraphInc{\mathcal{T}}$ with two
edges $e'_1, e'_2 \in E(\mathcal{T})$ of our choice. Choosing $e'_1
= e_1$ and $e'_2 \neq e_2$ yields two cycles
$\outHamilton{\mathcal{T}}, \inHamilton{\mathcal{T}}$ that meet the
requirements of Lemma~\ref{lemma:cycle_merging} and we can merge
them into one cycle $H$ for $G$.
\end{proof}

The proof of Theorem~\ref{theorem:triconnected} is constructive.
Embedding $G$ and identifying all separating triangles in $G$ can be
done in linear time. Augmenting a cycle and merging two of them
takes constant time. Disjointness of separating triangles yields a
linear number of subproblems and every edge occurs in at most one
such subproblem. Hence, the total time spent for the subroutine of
Lemma~\ref{lemma:linear_time_one_crossing} is linear in the size of
$G$.

\begin{corollary}
A subhamiltonian cycle of a triconnected 4-planar graph can be found
in linear time.
\end{corollary}

In this section, we have shown that in the triconnected case a
rather simple technique can be used to efficiently compute a
subhamiltonian cycle in a 4-planar graph. However, the property that
G is triconnected has been used extensively throughout this section,
thus, a relaxation to biconnectivity is not straightforward.

\section{Two-Page Book Embeddings of General 4-Planar Graphs}
\label{sec:general-planar}
In this section, we prove that any planar graph of maximum degree
four admits a two-page book embedding. The proof is given by a
recursive combinatorial construction, which determines the order of
the vertices along the spine and the page in which each edge is
drawn. W.l.o.g. we assume that the input graph $G$ is biconnected,
since it is known that the page number of a graph equals the maximum
of the page number of its biconnected components \cite{bk-btg-79}.
Note that one can neglect the exact geometry, as two edges that are
drawn on the same page cross if and only if their endpoints
alternate along the spine. We say that an edge $e$ \emph{nests a
vertex} $v$ iff one endpoint of $e$ is to the left of $v$ along the
spine and the other endpoint of $e$ to its right. We also say that
an edge $e$ \emph{nests an edge} $e'$ iff both $e$ and $e'$ are
drawn on the same page and both endpoints of $e'$ are nested by $e$.
Observe that nested edges do not cross.

The general idea of our algorithm is as follows: First remove from
$G$ cycle $C_{out}$ delimiting the outerface of $G$ and
\emph{contract} each bridge-block\footnote{The \emph{bridge-blocks}
of a connected graph $G$ are the connected components formed by
deleting all bridges of $G$. The bridge-blocks and the bridges of
$G$ have a natural tree structure, called \emph{bridge-block tree}.}
of the remaining graph into a single vertex. Let $F$ be the implied
graph, which is a forest in general, since $G-C_{out}$ is not
necessarily connected. Cycle $C_{out}$ is embedded, s.t.: (i)~the
order of the vertices of $C_{out}$ along the spine is fixed (and
follows the one in which the vertices of $C_{out}$ appear along
$C_{out}$), and, (ii)~all edges of $C_{out}$ are on the same page,
except for the one that connects its outermost vertices. Then, we
describe how to embed without crossings: (i)~the chords of
$C_{out}$, (ii)~forest $F$, and, (iii)~the edges between $C_{out}$
and $F$. To obtain a two-page book embedding of $G$, we replace each
vertex of $F$ with a cycle (embedded similarly to $C_{out}$), whose
length equals to the length of the cycle delimiting the outerface of
the bridge-block it corresponds to in $G-C_{out}$, and recursively
embed its interior.

More formally, consider an arbitrary simple cycle $C: v_1
\rightarrow v_2 \rightarrow \ldots \rightarrow v_k \rightarrow v_1$
of $G$. The removal of $C$ results in two planar subgraphs
$\inGraph{C}$ and $\outGraph{C}$ of $G$ that are the components of
$G-C$ that lie in the interior and exterior of $C$ in $G$, resp.
Note that $\inGraph{C}$ and $\outGraph{C}$ are not necessarily
connected. Let $\inGraphInc{C}$ ($\outGraphInc{C}$, resp.) be the
subgraph of $G$ induced by $C$ and $\inGraph{C}$ ($\outGraph{C}$,
resp.). For the recursive step, we assume the following invariant
properties:

\begin{enumerate}[{I}P-1:]
\item \label{ip:1} The order of the vertices of $\outGraphInc{C}$ along the spine $\ell$ is fixed and the page in which each edge of $\outGraphInc{C}$ is drawn (i.e., top or bottom) is determined s.t. the book embedding of $\outGraphInc{C}$ is planar. In other words, we assume that we have already produced a two-page book embedding for $\outGraphInc{C}$, in which no edge crosses the spine.

\item \label{ip:extra} The combinatorial embedding of $\outGraphInc{C}$ is consistent with a given planar combinatorial embedding of $G$.

\item \label{ip:2} The vertices of $C$ occupy consecutive positions along $\ell$, s.t. $v_1$ ($v_k$, resp.) is the leftmost (rightmost, resp.) along $\ell$. Moreover, all edges of $C$ are on the same page, except for the one that connects $v_1$ and $v_k$. Say w.l.o.g. that $(v_1,v_k)$ is on the top-page (or \emph{top-drawn}), while the remaining edges of $C$, namely edges $(v_i,v_{i+1})$ for $1\leq i<k$, are on the bottom-page (or \emph{bottom-drawn}); see Fig.\ref{fig:outerface}.

\item \label{ip:3} If $C$ is not identified with the cycle delimiting the outerface of $G$, the degree of either $v_1$ or $v_k$ is at most 3 in $\inGraphInc{C}$. Say w.l.o.g. that $v_k$ is of degree at most 3.

\item \label{ip:4} If vertex $v_1$ has degree 4 in $\inGraphInc{C}$, then it is adjacent to zero or two chords of $C$.
\end{enumerate}

\begin{figure}[ht]
    \centering
    \includegraphics[width=.4\textwidth]{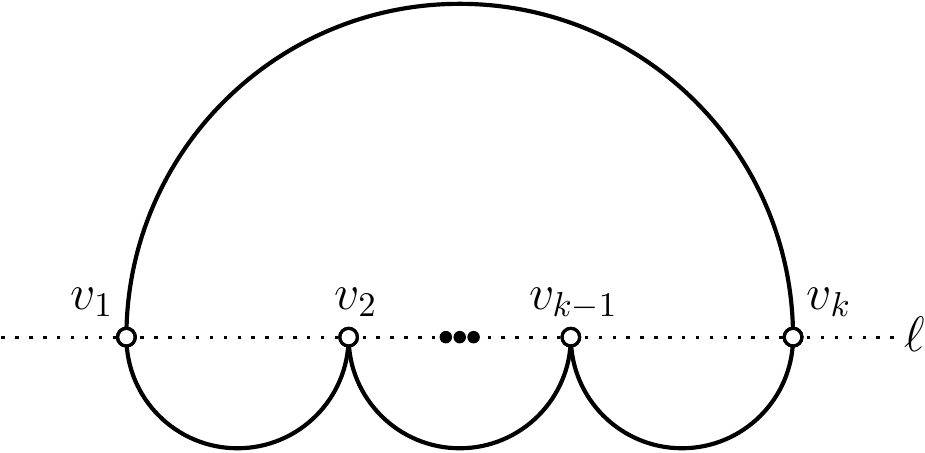}
    \caption{Illustration of invariant property IP-\ref{ip:2}.}
    \label{fig:outerface}
\end{figure}

We explicitly notice that the combinatorial embedding specified in
IP-\ref{ip:extra} is maintained throughout the whole drawing
process. This combined with the fact that every edge entirely lies
on one page (i.e., no edge crosses the spine; see IP-\ref{ip:1}) is
sufficient to ensure planarity. In the following, we describe in
detail how to recursively produce a two-page book embedding of
$\inGraphInc{C}$. Note that we first present the recursive step of
our algorithm and then its base, since this approach shows better
how the different ideas flow one after the other. Let $v_i$ be a
vertex of $C$, $i=1,\ldots,k$. Since $G$ is of max-degree $4$, $v_i$
is incident to at most two undrawn edges. Assume that $v_i$ has at
least one undrawn edge. We refer to the edge incident to $v_i$ that
follows $(v_i, v_{(i+1)\text{ mod }k})$ in the counterclockwise order
of the edges around $v_i$ (as defined by the combinatorial embedding
specified by IP-\ref{ip:extra}), as the \emph{right edge} of $v_i$.
If $v_i$ is adjacent to two undrawn edges, then the one that is not
identified with the right edge of $v_i$ is its \emph{left edge};
otherwise, the left and the right edge of $v_i$ are identified.

Initially, we draw the chords of $C$ on the top-page. By
IP-\ref{ip:extra} and IP-\ref{ip:2}, no two chords intersect. We
then draw $\inGraph{C}$ and the edges between $C$ and $\inGraph{C}$.
Note that $\inGraph{C}$ is not necessarily connected. Hence, its
bridge-block trees form a forest. As already stated, we contract
each bridge-block of $\inGraph{C}$ into a single vertex, which we
call \emph{\bv}; see
Figs.~\ref{fig:example_graph}-\ref{fig:example_forest}. We
distinguish two types of \bvs: those adjacent to vertices of $C$
(\emph{\ans}) and those adjacent to other \bvs only (\emph{\ccs}).
From the contraction, it follows that an edge between $C$ and a
certain \an can be of multiplicity at most two. Edges among \bvs are
always simple. We will first determine the positions of all \ans
along $\ell$. Consider an \an $c$, then among the edges between $c$ and
$C$, we select and \emph{mark} exactly one, s.t.: (i)~the marked
edge will be drawn on the bottom-page and (ii)~all other edges
incident to $c$ (i.e., either edges between $c$ and $C$ that are not
marked, or between $c$ and \bvs) will be drawn on the top-page. Let
$v_{l,c}$ be the leftmost vertex of $C$ adjacent to $c$ along
$\ell$. If $(c,v_{l,c})$ is simple, we select and mark this edge.
Otherwise, we mark the right edge of $v_{l,c}$. Hence, each \an has
exactly one marked edge (which we will shortly utilize to determine
its position along $\ell$) and each vertex of $C$ is incident to at
most two marked edges. Let $v\in C$ be a vertex of $C$ adjacent to
at least one \an through a marked edge. Then we have two cases:

\begin{figure}[t]
    \centering
    \begin{minipage}[b]{.32\textwidth}
        \centering
        \subfloat[\label{fig:example_graph}{Bridge-blocks of $\inGraph{C}$}]
        {\includegraphics[width=.8\textwidth,page=1]{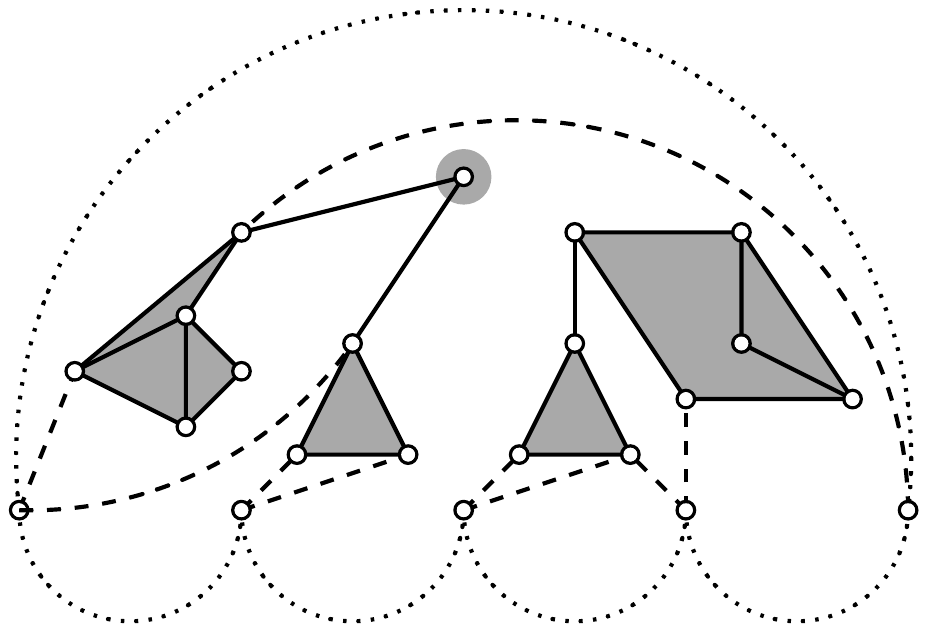}}
    \end{minipage}
    \hfill
    \begin{minipage}[b]{.32\textwidth}
        \centering
        \subfloat[\label{fig:example_forest}{Forest of \bvs}]
        {\includegraphics[width=.8\textwidth,page=2]{images/decomposition}}
    \end{minipage}
    \hfill
    \begin{minipage}[b]{.32\textwidth}
        \centering
        \subfloat[\label{fig:example_forest_solved}{Placement of \ans}]
        {\includegraphics[width=.8\textwidth,page=3]{images/decomposition}}
    \end{minipage}
    \begin{minipage}[b]{.19\textwidth}
        \centering
        \subfloat[\label{fig:right_placement_1}{}]
        {\includegraphics[width=.9\textwidth,page=1]{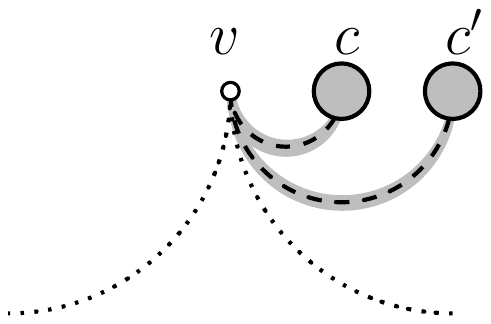}}
    \end{minipage}
    \hfill
    \begin{minipage}[b]{.19\textwidth}
        \centering
        \subfloat[\label{fig:right_placement_2}{}]
        {\includegraphics[width=.9\textwidth,page=2]{images/right_edge}}
    \end{minipage}
    \hfill
    \begin{minipage}[b]{.19\textwidth}
        \centering
        \subfloat[\label{fig:left_placement_1}{}]
        {\includegraphics[width=.9\textwidth,page=3]{images/right_edge}}
    \end{minipage}
    \hfill
    \begin{minipage}[b]{.19\textwidth}
        \centering
        \subfloat[\label{fig:right_placement_3}{}]
        {\includegraphics[width=.9\textwidth,page=4]{images/right_edge}}
    \end{minipage}
    \hfill
    \begin{minipage}[b]{.19\textwidth}
        \centering
        \subfloat[\label{fig:left_placement_2}{}]
        {\includegraphics[width=.9\textwidth,page=5]{images/right_edge}}
    \end{minipage}
    \caption{In all figures, the edges of $C$ are drawn dotted, bridge-blocks are colored gray and edges between $C$ and \ans are drawn dashed; marked edges are highlighted in gray.}
    \label{fig:example_potatoes}
\end{figure}

\begin{description}
\item[Case 1] \emph{$v$ is adjacent to exactly two \ans
$c$ and $c'$ through two marked edges $e$ and $e'$, resp.}: Assume
w.l.o.g. that $e$ is the left edge of $v$ and $e'$ its right edge.
Then, both $c$ and $c'$ are placed directly to the right of $v$ and
$c$ precedes $c'$ (see Fig.~\ref{fig:right_placement_1}). Note that
$v$ cannot be the rightmost vertex of $C$ due to IP-\ref{ip:3}.

\item[Case 2] \emph{$v$ is adjacent to one \an
$c$ through a marked edge $e$}: If $deg(v)=3$ in $\inGraphInc{C}$,
then we distinguish two sub-case. If $v$ is not the rightmost vertex
of $C$, then $c$ is placed directly to the right of $v$ (see
Fig.~\ref{fig:right_placement_2}). Otherwise, directly to its left
(see Fig.~\ref{fig:left_placement_1}). It now remains to consider
the case where $deg(v)=4$ in $\inGraphInc{C}$. In this case, by
IP-\ref{ip:3} it follows that $v$ is not the rightmost vertex of
$C$. Again, we distinguish two sub-cases:
\begin{itemize}
\item[$-$] \emph{If $e$ is the right edge of $v$}, then $c$ is placed directly to the right of $v$ (see Fig.~\ref{fig:right_placement_3}).

\item[$-$] \emph{If $e$ is the left edge of $v$}, then $c$ is placed directly to the left of $v$ (see Fig.~\ref{fig:left_placement_2}); $v$ cannot be  the leftmost vertex of $C$, as the right edge of $v$ would be a chord, violating IP-\ref{ip:4}.
\end{itemize}
\end{description}

As already stated, all marked edges are bottom-drawn. Edges between
\ans and $C$ that are not marked are top-drawn; see
Fig.~\ref{fig:example_forest_solved}. Observe that we do not change
the underlying combinatorial embedding of $G$, preserving
IP-\ref{ip:extra}. Hence, the book embedding constructed so far is
planar.

Before we proceed to describe how \ans ``determine'' the positions
of \ccs, we introduce the notion of (\emph{labeled}) \emph{anchored
tree} and investigate properties of it. Observe that \ccs form a new
forest (\emph{forest of \ccs}), which is subgraph of the initial
forest (containing all \bvs). Let $T$ be a tree of the forest of
\ccs and let $c_1,\dots,c_t$ be \ans that (i)~are adjacent to at
least one \cc of $T$, and (ii)~$c_i$ is to the left of $c_{i+1}$,
$i=1,\ldots,t-1$. We refer to $c_1,\dots,c_t$ as the \emph{\ans} of
$T$, and to the tree formed by $T$ and its \ans as the
\emph{anchored tree} of $T$, denoted by $\overline{T}$. We say that
two \ans of $\overline{T}$ are \emph{consecutive} iff there is no
\an of $\overline{T}$ between them (\ans that do not belong to
$\overline{T}$ or vertices of $C$ may lie in between).

\begin{lemma}\label{lem:basic}
For anchored trees the following hold:
\begin{inparaenum}[$(i)$]
\item \label{lem:common_anchors} Two trees $\overline{T}$ and $\overline{T'}$ share at most a common \an;
\item \label{lem:2-anchors} $\overline{T}$ contains at least two
\ans; and
\item \label{lem:leafs} every leaf of $\overline{T}$ is an \an of $\overline{T}$, and vice versa.
\end{inparaenum}
\end{lemma}
\begin{proof}
(\ref{lem:common_anchors})~If $\overline{T}$ and $\overline{T'}$
have two common \ans $c$ and $c'$, then there are two paths from $c$
to $c'$, one through $T$ and one through $T'$, which form a cycle of
\bvs, a contradiction. (\ref{lem:2-anchors})~If $\overline{T}$ has
no \ans, then no path from $C$ to $T$ exists, a contradiction since
$G$ is connected. If $\overline{T}$ has one \an $c$, then the edge
from $c$ to $T$ is a bridge, a contradiction since $G$ is
biconnected. Note that the edge from $c$ to $T$ is always simple;
double edges potentially occur between vertices of $C$ and \ans.
(\ref{lem:leafs})~Removing the \ans of $\overline{T}$, we obtain
$T$. If an \an of $\overline{T}$ is internal to $\overline{T}$, then
its removal disconnects $T$, a contradiction since $T$ is connected.
If there is a leaf $c \in \overline{T}$ that is not an \an of
$\overline{T}$, then the edge from $c$ to $T$ is a bridge, a
contradiction since $G$ is biconnected.
\end{proof}

Assume now that $\overline{T}$ is rooted at \an $c_1$ (\emph{rooted
anchored tree}). For an \an or \cc $c$ of $\overline{T}$, denote by
$p(c)$ the parent of $c$ in $\overline{T}$ and let $p(c_1)$ be any
of the vertices of $C$ adjacent to $c_1$. For an \cc $c$ of $T$
(i.e., non-leaf in $\overline{T}$), we define an order for its
children: if $c'$ and $c''$ are children of $c$, then $c'<c''$ iff
$c'$ precedes $c''$ in the counterclockwise order of the edges
around $c$ (defined by the combinatorial embedding specified by
IP-\ref{ip:extra}), when starting from $(c,p(c))$. By this order, we
label the vertices of $\overline{T}$ as they appear in the pre-order
traversal of $\overline{T}$ (\emph{labeled anchored tree}); see
Fig.\ref{fig:example_tree_ordered}.

\begin{figure}[t]
\centering
    \begin{minipage}[b]{.48\textwidth}
    \centering
    \subfloat[\label{fig:example_tree_ordered}{A labeled anchored tree $\overline{T}$}]
    {\includegraphics[width=\textwidth,page=1]{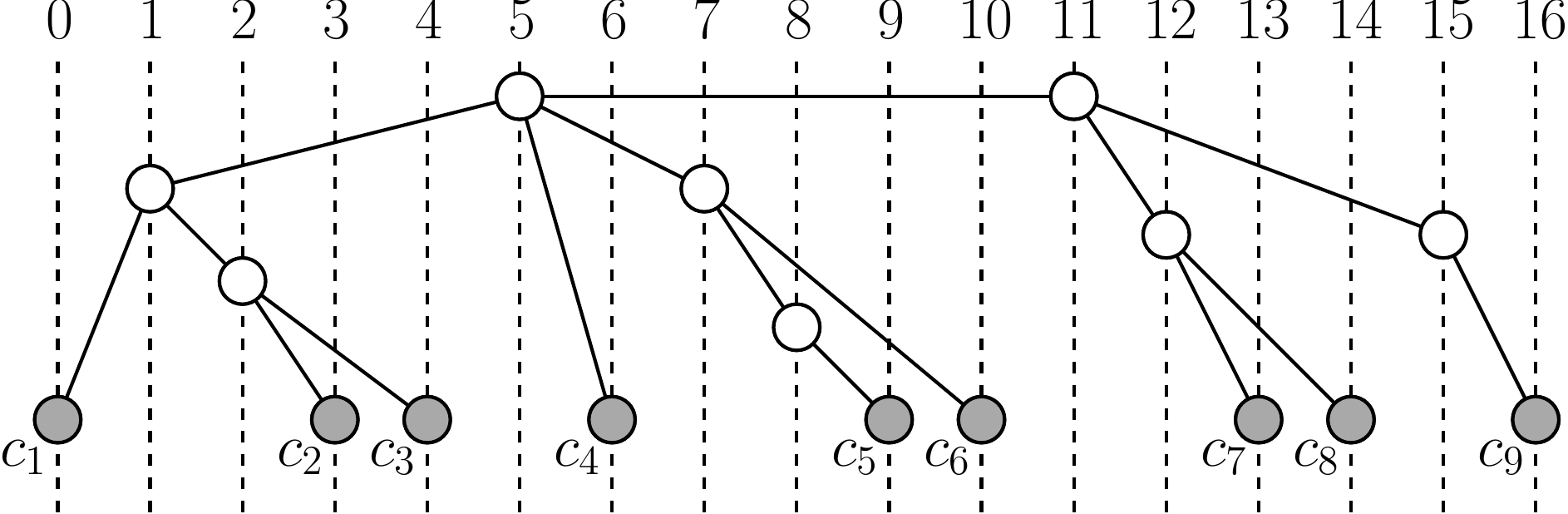}}
    \end{minipage}
    \hfill
    \begin{minipage}[b]{.48\textwidth}
    \centering
    \subfloat[\label{fig:example_tree_placed}{The placement of the \ccs of $\overline{T}$ among its \ans}]
    {\includegraphics[width=\textwidth,page=2]{images/tree_route}}
        \end{minipage}
    \caption{
    In both figures, \ans are colored gray; the indices of the vertical grid-lines denote the labeling of $\overline{T}$.}
    \label{fig:example_placement_tree}
\end{figure}

\begin{lemma}\label{lem:label_component}
For each \cc  $c$ of a labeled anchored tree $\overline{T}$ there is
(i)~at least an \an of $\overline{T}$ with label smaller than that
of $c$ and (ii)~at least another with label greater than that of
$c$.
\end{lemma}
\begin{proof}
The leftmost \an (i.e. root) of $\overline{T}$ is zero labeled,
which proves (i). The greatest labeled vertex of $\overline{T}$ is a
leaf of $\overline{T}$ (due to pre-order traversal) and by
Lemma~\ref{lem:basic}(\ref{lem:leafs}) an \an of $\overline{T}$.
This proves (ii).
\end{proof}

We first define the order in which the trees of the forest of \ccs
will be drawn. To do so, we create an auxiliary graph $G_{aux}^T$
whose vertices correspond to trees and there is a directed edge
$(v_{T'},v_T)$ in $G_{aux}^T$ iff $\overline{T'}$ has an \an between
two consecutive \ans of $\overline{T}$. The desired order is defined
by a topological sorting of $G_{aux}^T$, which always exists as the
following lemma suggests.

\begin{lemma}\label{lem:dag}
Auxiliary graph $G_{aux}^T$ is a directed acyclic graph.
\end{lemma}
\begin{proof}
Assume to the contrary that there is a cycle $v_{T_1} \rightarrow
\ldots v_{T_s} \rightarrow v_{T_1}$ in $G_{aux}^T$. Let $I_i$ be the
interval defined by the left/right-most \ans of $\overline{T_i}$.
Edge $(v_{T_i},v_{T_{i+1}})$ implies that there is an \an of
$\overline{T_i}$ between consecutive \ans of $\overline{T_{i+1}}$.
However, in this case all \ans of $\overline{T_i}$ should be between
the same two \ans of $\overline{T_{i+1}}$, as otherwise the
embedding specified by IP-\ref{ip:extra} is not planar. So,
$I_i\subseteq I_{i+1}$. By
Lemma~\ref{lem:basic}(\ref{lem:common_anchors}), it follows that
$I_i \neq I_{i+1}$. Hence, $I_1\subset\ldots\subset I_s\subset I_1$,
a contradiction.
\end{proof}

Lemma~\ref{lem:dag} implies that drawing the trees in the order
defined by a topological sorting of $G_{aux}^T$, assures that the
tree $T'$ will be drawn before $T$, if $\overline{T'}$ has an \an
that is between two consecutive \ans of $\overline{T}$ along $\ell$.
Now assume that we have drawn zero or more of these trees s.t.
(i)~all edges are top-drawn, (ii)~there are no edge crossings, and
(iii)~the combinatorial embedding specified by IP-\ref{ip:extra} is
preserved. Let $T$ be the next tree to be drawn. The following lemma
presents an important property of our drawing approach.

\begin{lemma}\label{lem:order}
Assume that all trees that precede $T$ in a topological sorting of
$G_{aux}^T$ have been drawn on the top-page without edge crossings
by preserving the combinatorial embedding specified by
IP-\ref{ip:extra}. If $e$ is a top-drawn edge that does not belong
to $\overline{T}$ and nests at least one \an of $\overline{T}$, then
it nests all \ans of $\overline{T}$.
\end{lemma}
\begin{proof}
If \emph{$e$ is the top-drawn edge of cycle $C$}, then $e$ nests all
\ans of $\overline{T}$, since all \ans of $\overline{T}$ are between
the left/right-most vertices of $C$. Now consider the case where $e$
is not the top-drawn edge of $C$. By
Lemma~\ref{lem:basic}(\ref{lem:2-anchors}), $\overline{T}$ has at
least two \ans, say $c$ and $c'$ (with $c$ to the left of $c'$), and
assume to the contrary that $e$ nests $c$ and not $c'$. If \emph{$e$
is an edge of an anchored tree $\overline{T'}$ drawn before
$\overline{T}$}, then by Lemma~\ref{lem:label_component}, both
endpoints of $e$ are between the left/right-most \ans of
$\overline{T'}$. Since $e$ nests $c$, $T$ should be drawn before
$T'$, a contradiction. Finally, if \emph{$e$ is not an edge of a
previously drawn anchored tree}, then each endpoint of $e$ is either
(i)~a vertex of $C$ or (ii)~an \an. Since such vertices are
connected to $C$ by a bottom edge, there is a path connecting the
endpoints of $e$ on the bottom-page, which together with $e$ forms a
cycle with $c$ in its interior and $c'$ on its exterior. Hence, the
embedding specified by IP-\ref{ip:extra} is not planar, a
contradiction.
\end{proof}

We now describe how to draw $T$ on the top page s.t. (i)~there are
no edge crossings, and, (ii)~the combinatorial embedding specified
by IP-\ref{ip:extra} is preserved. More precisely, we place each \cc
$c$ of $T$ between a pair of consecutive \ans of $\overline{T}$,
s.t. the label of $c$ is larger (smaller) than the label of the \an
to its left (right)\footnote{Note that the existence of this pair of
consecutive \ans of $\overline{T}$ is implied by
Lemma~\ref{lem:label_component}; since for each \cc $c$ of a labeled
anchored tree $\overline{T}$ there exist at least an \an of
$\overline{T}$ with label smaller than that of $c$ and at least
another with label greater than that of $c$, there should be two
consecutive ones with this property as well.}; for \ccs placed
between the same pair of \ans, the one with smaller label is to the
left; all edges of $\overline{T}$ are top-drawn (see
Fig.\ref{fig:example_tree_placed}). Note that we have not fully
specified the exact positions of the \ccs of $\overline{T}$ along
$\ell$, since between consecutive \ans of $\overline{T}$
there may exist \ans that do not belong to $\overline{T}$ or vertices of
$C$ or \ans/\ccs of trees that have already been drawn. Details will
be given shortly. Notice that all \ccs of $\overline{T}$ are placed
between its left/right-most \ans, which by Lemma~\ref{lem:order}
implies that if a top-drawn edge (that does not belong to
$\overline{T}$) nests at least one \an of $\overline{T}$, then it
nests the entire tree $\overline{T}$. By exploiting the
correspondence between the left-to-right order of the vertices of
$\overline{T}$ along $\ell$ and the labeling of $\overline{T}$, we
can prove that the drawing of $\overline{T}$ is planar.

\begin{lemma}\label{lem:tree_planar}
The drawing of the anchored tree $\overline{T}$ is planar.
\end{lemma}
\begin{proof}
Assume to the contrary that $e=(c_1, c_2)$ and $e'=(c'_1,c'_2)$ of
$\overline{T}$ cross. Since $e$ and $e'$ are top-drawn, their
endpoints alternate along $\ell$. Let the order on $\ell$ be
$c_1\rightarrow c'_1\rightarrow c_2\rightarrow c'_2$. Hence, $c_1$
is the parent of $c_2$, as the label of $c_1$ is smaller than that
of $c_2$ and they are adjacent in $\overline{T}$. Similarly, $c'_1$
is the parent of $c'_2$. Since between $c_1$ and $c_2$ are drawn
subtrees of $\overline{T}$ rooted at children of $c_1$ other than
$c_2$, it follows that $c'_1$ and $c'_2$ belong to a subtree rooted
at a child of $c_1$, different from $c_2$, which implies that the
label of $c'_2$ is smaller than that of $c_2$, a contradiction.
\end{proof}

Recall that we have not fully specified the exact positions of the
\ccs of $\overline{T}$ along $\ell$. Consider the following
scenario. There is a path $P$ of top-drawn edges (e.g., non-marked
edges incident to $C$ and/or edges of previously drawn trees)
joining a pair of consecutive \ans of $\overline{T}$ and our
algorithm must place an \cc $c$ of $\overline{T}$ between them.
Since $c$ is nested by an edge of $P$ and all edges of
$\overline{T}$ are top-drawn, an edge connecting $c$ with an \cc of
$\overline{T}$ placed between another pair of consecutive \ans of
$\overline{T}$ will cross $P$. The following lemma ensures that this
scenario cannot occur, as such a path cannot exist.

\begin{lemma}\label{lem:anchors}
Let $u_0, u_1,\ldots, u_{l+1}$, $l \geq 0$, be vertices (\ans/\ccs
are treated as vertices) drawn on $\ell$ from left to right, s.t.
$u_0$ and $u_{l+1}$ are two consecutive \ans of $\overline{T}$.
Assume that all trees anchored at $u_1,\ldots, u_l$ have been drawn
on the top-page without edge crossings by preserving the
combinatorial embedding specified by IP-\ref{ip:extra}, while $T$
has not been drawn. Then, there is an index $i \in
\{0,1,\ldots,l\}$, such that no two adjacent vertices $u_k$ and
$u_m$ exist with $0 \leq k \leq i$, $i+1 \leq m \leq l+1$ and
$(u_k,u_m)$ is top-drawn.
\end{lemma}
\begin{proof}
Since all trees anchored at $u_1,\ldots, u_l$ have been drawn, edges
incident to $u_1,\ldots, u_l$ are present in the drawing. For a
proof by contradiction, we make the following assumption: For all $i \in
\{0,\ldots,l\}$, there are two adjacent vertices $u_k$ and $u_m$
with $0 \leq k \leq i$, $i+1 \leq m \leq l+1$ and $(u_k,u_m)$ is on
the top-page. We first prove that there is a top-drawn path $P(u_0
\rightarrow u_{l+1}):~u_0\rightarrow u_{j_1} \ldots u_{j_p}
\rightarrow u_{l+1}$ consisting of vertices of
$\{u_0,\ldots,u_{l+1}\}$, whose edges are top-drawn and for each
edge of $P(u_0 \rightarrow u_{l+1})$ there is not a top-drawn edge
with endpoints in $\{u_0,\ldots,u_{l+1}\}$ that nests it. The
existence of $P(u_0 \rightarrow u_{l+1})$ will imply the desired
contradiction.

For $i=0$, by our assumption it follows that for some $m \in
\{1,\ldots,l+1\}$, edge $(u_0,u_m)$ exists and is on the top-page.
Let $j_1 \in \{1,\ldots,l+1\}$ be the maximum s.t. $(u_0,u_{j_1})$
is drawn on the top-page. If $j_1 = l+1$, then $P(u_0 \rightarrow
u_{l+1})$ exists. Let $j_1\neq l+1$. For $i=j_1$, it follows that
for some $k \in \{0,\ldots,j_1\}$ and $m \in \{j_1+1,\ldots,l+1\}$,
$(u_k,u_m)$ exists and is on the top-page. $k \notin
\{1,\ldots,j_1-1\}$, since otherwise $(u_0,u_{j_1})$ and $(u_k,u_m)$
would cross, which is not possible since the combinatorial embedding
specified by IP-\ref{ip:extra} is planar. Also, $k \neq 0$, since
$j_1$ was the maximum of $\{1,\ldots,l+1\}$, s.t. $(u_0,u_{j_1})$ is
drawn on the top-page. Hence, $k=j_1$. Let $j_2 \in
\{j_1+1,\ldots,l+1\}$ be the maximum, s.t. $(u_{j_1},u_{j_2})$ is
drawn on the top-page, and proceed as in the case $i=0$. This
procedure will eventually lead to $P(u_0\rightarrow u_{l+1})$. We
claim that $P(u_0\rightarrow u_{l+1})$ has at least one vertex of
$C$. Assume to the contrary that $P(u_0\rightarrow u_{l+1})$
contains only \ans/\ccs, which cannot belong to $\overline{T}$,
since $u_0$ and $u_{l+1}$ are consecutive \ans of $\overline{T}$. By
Lemma~\ref{lem:basic}(\ref{lem:leafs}), $u_0$ and $u_{l+1}$ are
leaves of $\overline{T}$.  Hence, the path from $u_0$ to $u_{l+1}$
through $T$ and $P(u_0\rightarrow u_{l+1})$ form a cycle of
\ans/\ccs, a contradiction. Let $u$ be the rightmost vertex of $C$
in $P(u_0\rightarrow u_{l+1})$ and $c$ be the neighbor of $u$ in
$P(u_0\rightarrow u_{l+1})$ to the right of $u$ on $\ell$. Since
$u_{l+1}$ is an \an of $\overline{T}$, $c$ is well-defined and is
either an \an or an \cc. Now observe that $c$ is adjacent to $u$ and
$u\in C$, which implies that $c$ is an \an and hence is incident to
a marked edge, say $(v,c)$, where $v\in C$ ($u=v$ is possible).
Assume that $u \neq v$. Then, $v$ is the leftmost neighbor of $c$,
which suggests that the order on $\ell$ is: $v\rightarrow
u\rightarrow c$. However, such an order cannot occur since $(v,c)$
is marked and $u$ in between. It follows that $u=v$. Since $(u,c)\in
P(u_0\rightarrow u_{l+1})$ (i.e. top-drawn) and is marked (i.e.
bottom-drawn), $(u,c)$ is double edge. Now observe that $u \in C$
has two incident edges on $C$, which contribute 2 to its degree.
Double edge $(u,c)$ also contributes 2. Up to now $deg(u)=4$. The
contradiction follows from $u$'s additional edge in
$P(u_0\rightarrow u_{l+1})$.
\end{proof}

We are now ready to specify the exact positions of the \ccs of
$\overline{T}$ along $\ell$. Recall that the \ans of $\overline{T}$
are denoted by $c_i$, $i=1,\ldots,t$, s.t. $c_i$ is to the left of
$c_{i+1}$. Now assume that a particular number of \ccs of
$\overline{T}$ should be drawn between two consecutive \ans $c_i$
and $c_{i+1}$ of $\overline{T}$, $i=1,\ldots,t-1$. By
Lemma~\ref{lem:anchors}, there is a pair of vertices that are
between $c_i$ and $c_{i+1}$ along $\ell$ and there is not a
top-drawn edge with endpoints between $c_i$ and $c_{i+1}$ nesting
both of these vertices. We benefit from this by placing between this
particular pair of vertices all \ccs of $\overline{T}$ that must
reside between $c_i$ and $c_{i+1}$. Their relative order is not
affected, i.e., for \ccs placed between $c_i$ and $c_{i+1}$, the one
with smaller label is to the left. Lemma~\ref{lem:tree_planar}
ensures the planarity of $\overline{T}$. It remains to prove that
the combinatorial embedding specified by IP-\ref{ip:extra} is
preserved.

\begin{lemma}\label{lem:embedding}
Assume that all trees that precede $T$ in a topological sorting of
$G_{aux}^T$ have been drawn on the top-page without edge crossings
by preserving the combinatorial embedding specified by
IP-\ref{ip:extra}. When $\overline{T}$ is drawn, the combinatorial
embedding specified by IP-\ref{ip:extra} is also preserved.
\end{lemma}
\begin{proof}
Since the drawing of $\overline{T}$ preserves the order of the edges
of $\overline{T}$ around all \ccs, the combinatorial embedding
specified by IP-\ref{ip:extra} is preserved for all \ccs of
$\overline{T}$. We will prove that the lemma holds for all \ans of
$\overline{T}$. Let $c$ be an \an of $\overline{T}$ and denote by
$e_c$ the marked edge incident to $c$ (which is bottom-drawn). Let
also $e_p$ and $e_t$ be two edges incident to $c$ s.t. $e_p$ is an
edge among those drawn before $T$ and $e_t$ is an edge of
$\overline{T}$. We restrict our proof to the case where in the
combinatorial embedding specified by IP-\ref{ip:extra},  $e_p$
precedes $e_t$ in the clockwise traversal of the edges around $c$
when starting from $e_c$ and $c$ is the left endpoint of $e_p$ along
$\ell$. The remaining cases are treated similarly. Then, there is a
simple path of drawn edges (other than $e_p$) that joins the two
endpoints of $e_p$ and together with $e_p$ forms a face of $G$. Let
$C(e_p)$ be the cycle bounding this face. Since $e_p$ precedes $e_t$
in the clockwise traversal of the edges around $c$ when starting
from $e_c$, $T$ lies in the interior of $C(e_p)$. Hence, there is a
top-drawn edge that belongs to $C(e_p)$ (possibly edge $e_p$) that
does not belong to $\overline{T}$ and that nests an \an of
$\overline{T}$. By Lemma~\ref{lem:order}, this edge nests all \ans
of $\overline{T}$ (including \an $c$). Since $c$ belongs to
$C(e_p)$, it follows that the only edge of $C(e_p)$ that nests
$\overline{T}$ is edge $e_p$. Thus, $c$ is the leftmost \an of
$\overline{T}$ and the entire drawing of $\overline{T}$ is nested by
$e_p$. After drawing $\overline{T}$, $e_p$ still precedes $e_t$ in
the clockwise traversal of the edges around $c$ when starting from
$e_c$, as desired.
\end{proof}

\begin{figure}[t]
\centering
    \begin{minipage}[b]{.17\textwidth}
    \centering
    \subfloat[\label{fig:cutvertex_bad_case}{}]
    {\includegraphics[width=\textwidth,page=1]{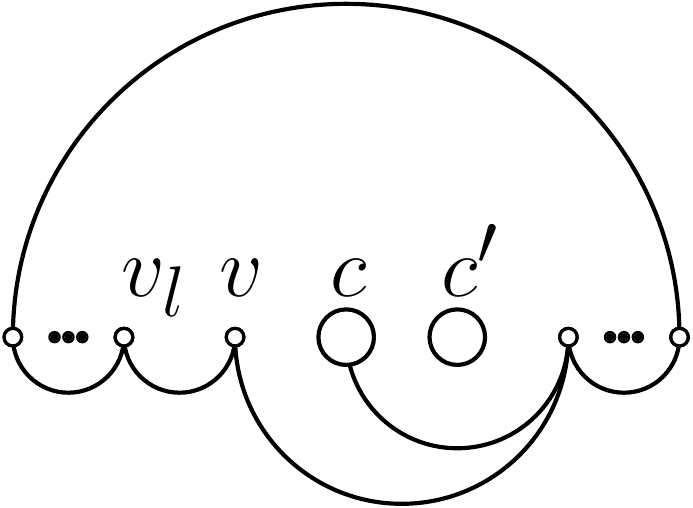}}
    \end{minipage}
    \hfill
    \begin{minipage}[b]{.36\textwidth}
    \centering
    \subfloat[\label{fig:cutvertex_push}{}]
    {\includegraphics[width=\textwidth,page=2]{images/cutvertex}}
    \end{minipage}
    \hfill
    \begin{minipage}[b]{.45\textwidth}
    \centering
    \subfloat[\label{fig:cutvertex_solved}{}]
    {\includegraphics[width=\textwidth,page=3]{images/cutvertex}}
    \end{minipage}
    \caption{Configuration considered in Lemma~\ref{lem:cutvertices}:
    (a) A situation in which placing \an $c$ to the left of $v$ creates
    crossings.
    (b) Moving \bv $c'$ to the right of $v_r$.
    (c) Edge $(v,v_r)$ can be drawn on the top half-plane.}
    \label{fig:lemma_cutvertices}
\end{figure}

In the following lemma, we turn our attention to the case where $C$
contains a vertex of degree $2$ in $\inGraphInc{C}$ (other than its
leftmost or rightmost vertex). We will utilize this lemma later.

\begin{lemma}\label{lem:cutvertices}
Let $v$ be a vertex of $C$ with degree 2 in $\inGraphInc{C}$ that is
not the left/right-most vertex of $C$. Let also $v_r$ ($v_l$) be its
next neighbor on $C$ to its right (left resp.). Since edge $(v,v_r)$
belongs to $C$, it is drawn on the bottom-page. However, it can also
be drawn on the top-page without edge-crossings, while the
combinatorial embedding specified by IP-\ref{ip:extra} is
maintained.
\end{lemma}
\begin{proof}
If no \bv is drawn between $v$ and $v_r$, then obviously $(v,v_r)$
can be drawn on the top-page. Otherwise, we will move the \bvs
in between to the left of $v$, so that $v$ and $v_r$ are consecutive
along $\ell$. This is not possible if there is an \an $c$ between
$v$ and $v_r$ s.t. $(c,v_r)$ is bottom-drawn (see
Fig.~\ref{fig:cutvertex_bad_case}). Alternatively, we could place
$v$ between $c$ and $v_r$. However, in this case $(v_r,c)$ and
$(v_l,v)$ cross. We could overcome this problem if $(v_r,c)$ is
redrawn on the top-page. This is not possible if there is a \bv $c'$
between $c$ and $v_r$. We have two cases:
\begin{inparaenum}[$(i)$]
\item \emph{$c'$ is an \an, i.e., adjacent to a vertex of $C$}. Then $c'$ can only
be adjacent to $v_r$ through a marked edge. Hence, $c$ and $c'$ are
two \ans that are both to the left of $v_r$ and adjacent to $v_r$
through marked edges, which is not valid by the algorithm, a
contradiction. \item \emph{$c'$ is an \cc}. Then $c'$ belongs to a
tree $T$. All \ccs of $T$ are placed between the left/right-most
\ans of $\overline{T}$. Let $u$ and $u'$ be consecutive \ans of
$\overline{T}$, ordered on $\ell$ as $u\rightarrow c'\rightarrow
u'$; $u=c$ is possible (see the left part of
Fig.~\ref{fig:cutvertex_push}). However, $u'$ cannot be between $c$
and $v_r$ (otherwise the previous case applies for $u'$), thus, $u'$ is
to the right of $v_r$. We claim that Lemma~\ref{lem:anchors} holds
for $u_0=v_r$ and $u_{l+1}=u'$, even though $u_0$ is not an \an but
a vertex of $C$ (the detailed proof is similar to the one of
Lemma~\ref{lem:anchors}). Hence, there are two consecutive vertices
between $v_r$ and $u'$ s.t. $c'$ can be placed between them (and not
between $c$ and $v_r$); see Fig.~\ref{fig:cutvertex_push}. The same
holds for every \cc that was initially placed between $c$ and $v_r$.
If we move all \ccs between $v_r$ and $u'$ by keeping their relative
order unchanged, then $(c,v_r)$ can be drawn on the top-page, and
the problem is resolved (see Fig.~\ref{fig:cutvertex_solved}).
\end{inparaenum}
\end{proof}

Up to now, we have drawn $\inGraphInc{C}$, s.t., every bridge-block
of $\inGraph{C}$ is contracted to a \bv that lies on $\ell$ and each
edge is drawn either on the bottom (if it is a marked edge) or on
the top-page (otherwise). Also, we preserved the order of the
vertices of $C$ on $\ell$ and the embedding of $G$ specified by
IP-\ref{ip:extra}. Hence, crossings in $\outGraphInc{C}$ cannot
occur. Next, we describe how to recursively proceed. Let $c$ be a
\bv of $\inGraph{C}$ with outerface $\mathcal{F}_c$. Initially,
assume that $\mathcal{F}_c$ is a simple cycle. If $c$ is an \an,
denote by $w_0$ the vertex of $\mathcal{F}_c$ incident to the marked
edge of $c$. If $c$ is an \cc, then $c$ belongs to an anchored tree.
In this case, $w_0$ denotes the vertex of $\mathcal{F}_c$ adjacent
to the closest neighbor of $c$ to its left, which is well-defined
since $c$ is always placed between two consecutive \ans of the
anchored tree it belongs to. Let $w_0,w_1,\dots,w_m$ be the vertices
of $\mathcal{F}_c$, in the clockwise traversal of $\mathcal{F}_c$
from $w_0$ (see Fig.~\ref{fig:node_explode_before}).

\begin{figure}[t]
\centering
    \begin{minipage}[b]{.27\textwidth}
        \centering
        \subfloat[\label{fig:node_explode_before}{}]
        {\includegraphics[width=\textwidth,page=1]{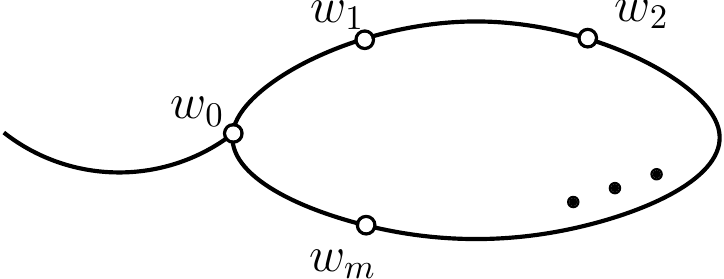}}
    \end{minipage}
    \hfill
    \begin{minipage}[b]{.27\textwidth}
        \centering
        \subfloat[\label{fig:node_explode_1}{}]
        {\includegraphics[width=.6\textwidth,page=2]{images/node_explode}}
    \end{minipage}
    \hfill
    \begin{minipage}[b]{.27\textwidth}
        \centering
        \subfloat[\label{fig:node_explode_2}{}]
        {\includegraphics[width=.7\textwidth,page=3]{images/node_explode}}
    \end{minipage}
    \begin{minipage}[b]{.4\textwidth}
        \centering
        \subfloat[\label{fig:tree_problem}{}]
        {\includegraphics[width=\textwidth,page=4]{images/node_explode}}
    \end{minipage}
    \hfill
    \begin{minipage}[b]{.4\textwidth}
        \centering
        \subfloat[\label{fig:tree_solution}{}]
        {\includegraphics[width=\textwidth,page=5]{images/node_explode}}
    \end{minipage}
        \caption{(a) The outerface of a \bv  $c$.
        (b)-(c) different cases that occur when drawing the outerface of $c$, in the case where $c$ is \an.
        (d)~\Cc $c$ needs to be repositioned.
        (e)~Its placement is determined by Lemma~\ref{lem:tree_fix}.}
    \label{fig:node_explode}
\end{figure}

First assume that $c$ is an \an, i.e., $w_0$ is incident to a marked
edge. We place the vertices of $\mathcal{F}_c$ on $\ell$ as follows:
(i)~$w_0$ occupies the position of $c$ and it is the rightmost
vertex of $\mathcal{F}_c$ on $\ell$, (ii)~$w_1$ is the leftmost
vertex of $\mathcal{F}_c$ on $\ell$, (iii)~$w_i$ is to the left of
$w_{i+1}$ for $i=1,\ldots,m-1$, and, (iv)~there are no vertices in between; 
see Fig.~\ref{fig:node_explode_1}. All edges of
$\mathcal{F}_c$ are top-drawn, except for $(w_1,w_0)$. This
placement is always feasible, except for the case in which in the
combinatorial embedding specified by IP-\ref{ip:extra} there is an
edge incident to $w_0$ that is between $(w_0,w_1)$ and the marked
edge incident to $w_0$ in the counterclockwise order of the edges
around $w_0$ when starting from $(w_0,w_1)$; see
Fig.~\ref{fig:node_explode_2}. In this case, we place $w_0$ to the
left of $w_1,\dots,w_m$, s.t. $w_0$ is the leftmost vertex of
$\mathcal{F}_c$. So, $(w_0,w_m)$ is the bottom-drawn edge of
$\mathcal{F}_c$.

Suppose now that $c$ is an \cc. Let $w$ be the closest neighbor of
$c$ to its left on $\ell$. Then, $w$ is the parent of $c$ in the
tree in which $c$ belongs to and $(w_0,w)$ is top-drawn. We place
the vertices of $\mathcal{F}_c$ as follows: (i)~$w_0$ occupies the
position of $c$ and it is the leftmost vertex of $\mathcal{F}_c$ on
$\ell$, (ii)~$w_m$ is the rightmost vertex of $\mathcal{F}_c$ on
$\ell$, (iii)~$w_i$ is to the left of $w_{i+1}$, $i=1,\ldots,m-1$,
and, (iv)~there are no vertices in between. All edges of
$\mathcal{F}_c$ are top-drawn, except for $(w_0,w_m)$. This
placement is infeasible only when in the combinatorial embedding
specified by IP-\ref{ip:extra} there is an edge incident to $w_0$,
say $(w_0,w')$, and between  $(w_0,w_m)$ and $(w_0,w)$ in the
clockwise order of the edges around $w_0$ when starting from
$(w_0,w_m)$  (see Fig.~\ref{fig:tree_problem}). In this case,
$(w_0,w')$ cannot be drawn on the top-page, as required for edges
incident to \ccs. More precisely, since $c$ has only its parent to
its left among the \bvs of the anchored tree it belongs to, it
follows that, $w'$ is to the right of $c$. Hence, $(w_0,w')$ cannot
be drawn on the top-page, without deviating the combinatorial
embedding specified by IP-\ref{ip:extra}. Since $G$ is biconnected,
$c$ is adjacent to at least another \bv, say $w''$, s.t. $w'' \notin
\{w,w'\}$. The following lemma takes care of this case.

\begin{lemma}\label{lem:tree_fix}
\Cc $c$ can be repositioned on $\ell$, s.t.: (i)~ $c$ is placed
between two consecutive \ans of $\overline{T}$. (ii)~The
combinatorial embedding specified by IP-\ref{ip:extra} is preserved
and the edges $(w_0,w)$, $(w_0,w')$ and $(c,w'')$ are top-drawn and
crossing-free. (iii)~$w_0$ is leftmost vertex of $\mathcal{F}_c$ and
$w_i$ is to the left of $w_{i+1}$, $i=1,\ldots,m-1$; All edges of
$\mathcal{F}_c$ are top-drawn, except for $(w_0,w_m)$.
\end{lemma}
\begin{proof}
$w$ is the parent of $c$ and $w'$, $w''$ are children of $c$ in
$\overline{T}$, with $w'$ being the first child of $c$. For our
proof, $w''$ is its second child. So, $(c,w)$, $(c,w')$ and
$(c,w'')$ are consecutive around $c$ as in
Fig.~\ref{fig:tree_problem}. Let $\overline{T(w')}$ and
$\overline{T(w'')}$ be subtrees of $\overline{T}$ rooted at $w'$ and
$w''$, resp. Initially, $c$ is to the left of all vertices of
$\overline{T(w')}$, all vertices of $\overline{T(w')}$ are to the
left of all vertices of $\overline{T(w'')}$ and there are no \ccs of
$\overline{T}$ in between. We place $c$ between the rightmost
(leftmost) \an  of $\overline{T(w')}$ ($\overline{T(w'')}$); see
Fig.~\ref{fig:tree_solution}. So, $c$ is placed between two
consecutive \ans of $\overline{T}$. If we place the vertices of
$\mathcal{F}_c$, with $w_0$ being leftmost on $\mathcal{F}_c$ and
$w_i$ to the left of $w_{i+1}$, then $(w_0,w)$, $(w_0,w')$ and
$(c,w'')$ are drawn on the top-page and the embedding is preserved.
\end{proof}

If we process all \ccs that have to be repositioned from right to
left along $\ell$, then by Lemma~\ref{lem:tree_fix} we obtain a
planar drawing in which the embedding specified by IP-\ref{ip:extra}
is preserved once\begin{wrapfigure}{r}{.31\textwidth}
\includegraphics[width=.31\textwidth,page=6]{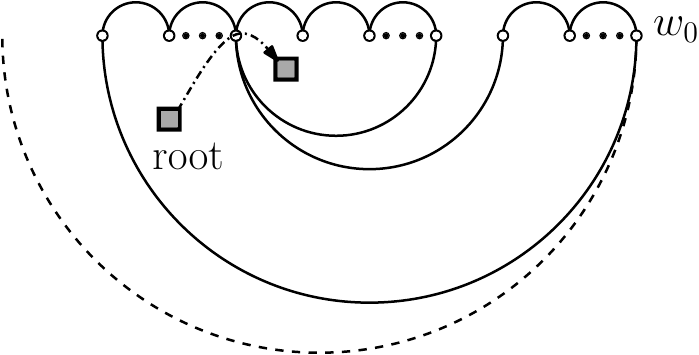}
\caption{$\mathcal{F}_c$ is not simple}
\label{fig:node_explode_cutvertex}
\end{wrapfigure} the outerface of
 each \bv is drawn and all edges
that connect \bvs are eventually drawn on the top-page. Initially,
we assumed that $\mathcal{F}_c$ is simple. If not so,
$\mathcal{F}_c$ consists of smaller simple \emph{subcycles}, s.t.
(i)~any two subcycles share at most one vertex of $\mathcal{F}_c$
and (ii)~any vertex of $\mathcal{F}_c$ is incident to at most two
subcycles. Hence, the \emph{``tangency graph''} of these subcycles
(which has a vertex for each subcycle and an edge between every pair
of subcycles that share a vertex) is a tree. Define  $w_0$ as in the
case of simple cycle and let the tangency tree be rooted at the
cycle containing $w_0$. Due to degree restriction, $w_0$ cannot be
incident to two subcycles. We draw the subcycles of $\mathcal{F}_c$
in the order implied by the Breadth First Search (BFS) traversal of
the  tangency tree. The first one (incident to $w_0$) is drawn as in
the case of simple cycle. Each next subcycle is plugged into the
drawing, as shown in Fig.~\ref{fig:node_explode_cutvertex}.

It remains to ensure that IP-\ref{ip:1} up to IP-\ref{ip:4} are
satisfied when a simple cycle, say $C_s$, is recursively drawn.
IP-\ref{ip:1} holds, since each edge is drawn either on the bottom
(if it is a marked edge) or on the top-page (otherwise) and no two
edges intersect. Lemma~\ref{lem:embedding} implies
IP-\ref{ip:extra}. If $C_s$ is the outerface of a \bv or a leaf in
the tangency tree, then IP-\ref{ip:2} trivially holds. If $C_s$ is a
non-leaf in the tangency tree, it contains at least one edge on the
bottom-page (see Fig.~\ref{fig:node_explode_cutvertex}). This
violates IP-\ref{ip:2}. However, we can benefit from
Lemma~\ref{lem:cutvertices} since the edge which is improperly
bottom-drawn is incident to a vertex (of degree four) that is not
adjacent to any other vertex in the interior of $C_s$. For the sake
of the recursion we assume that it is drawn on the top-page and once
$C_s$ is completely drawn, we redraw it on the bottom-page using
Lemma~\ref{lem:cutvertices}. If $C_s$ is the outerface of a \bv or
root of the tangency tree of a non-simple outerface $\mathcal{F}_c$,
then at least one vertex of $C_s$ is adjacent to $\outGraph{C_s}$.
If $C_s$ is an internal node of the tangency tree of
$\mathcal{F}_c$, then its leftmost vertex has two edges in
$\outGraph{C_s}$. Hence, IP-\ref{ip:3} also holds. 

However, IP-\ref{ip:4} does not necessarily hold. To cope with this
case, consider a simple cycle $C_s$ and, with a slight abuse of the
notation developed so far, denote by $w_1,\dots,w_m$ the vertices of
$C_s$ from left to right along $\ell$. If IP-\ref{ip:4} is violated,
then $deg(w_1)=4$ in $\inGraphInc{C_s}$ and $w_1$ is incident to
exactly one chord of $C_s$, say $(w_1,w_i)$,
$i\in\left\{3,\ldots,m-1\right\}$; see
Fig.~\ref{fig:ip4_before_general}. Let $v$ be the other neighbor of
$w_1$ in $\inGraph{C_s}$. Clearly, $v\notin C_s$. In general,
$(w_1,w_i)$ belongs to a path of chords stemming from $w_1$. Let
$w_j$, $j\geq i$, be the end of this path be the end of this path
$P(w_1\rightarrow w_j)$. The degree restriction implies that
$P(w_1\rightarrow w_j)$ is uniquely defined. We refer to it as the
\emph{separating path of chords} of $C_s$, since it splits
$\inGraphInc{C_s}$ into two subgraphs (see
Fig.~\ref{fig:ip4_split_left}-\ref{fig:ip4_split_right}):

\begin{itemize}
\item[$-$] $\inGraphInc{C_l}$ with outerface $C_l$
consisting of the edges $(w_1,w_2)$, $(w_2,w_3)$, $\ldots$,
$(w_{j-1},w_j)$ and the edges of $P(w_1\rightarrow w_j)$
(highlighted in gray in Fig.~\ref{fig:ip4_before_general}) and

\item[$-$] $\inGraphInc{C_r}$ with outerface $C_r$
consisting of the edges $(w_j,w_{j+1})$, $\ldots$, $(w_{m-1},w_m)$,
$(w_m,w_1)$ and the edges of $P(w_1\rightarrow w_j)$.
\end{itemize}

\begin{figure}[t!]
\centering
    \begin{minipage}[b]{.37\textwidth}
        \centering
        \subfloat[\label{fig:ip4_before_general}{A graph $\inGraphInc{C_s}$ which violates IP-\ref{ip:4}.}]
        {\includegraphics[width=\textwidth,page=1]{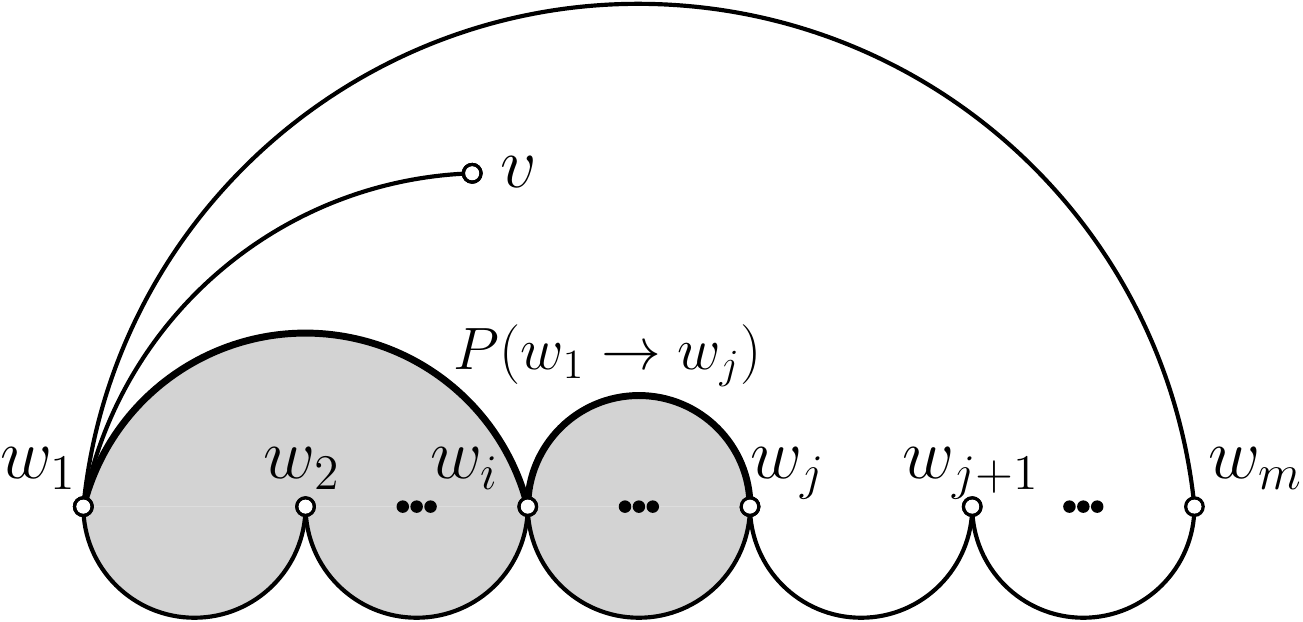}}
    \end{minipage}
    \hfill
    \begin{minipage}[b]{.24\textwidth}
        \centering
        \subfloat[\label{fig:ip4_split_left}{Subgraph $\inGraphInc{C_l}$.}]
        {\includegraphics[width=\textwidth,page=2]{images/ip4}}
    \end{minipage}
    \hfill
    \begin{minipage}[b]{.37\textwidth}
        \centering
        \subfloat[\label{fig:ip4_split_right}{Subgraph $\inGraphInc{C_r}$.}]
        {\includegraphics[width=\textwidth,page=3]{images/ip4}}
    \end{minipage}
    \hfill
    \begin{minipage}[b]{.37\textwidth}
        \centering
        \subfloat[\label{fig:ip4_right_1}{$\inGraphInc{C_r}$ in the case where $j=m$. }]
        {\includegraphics[width=\textwidth,page=4]{images/ip4}}
    \end{minipage}
    \hfill
    \begin{minipage}[b]{.61\textwidth}
        \centering
        \subfloat[\label{fig:ip4_right_2}{Graph $\inGraphInc{C_r'}$ obtained from $\inGraphInc{C_r}$ by contracting $P(w_1\rightarrow w_j)$ into a vertex}]
        {\includegraphics[width=\textwidth,page=5]{images/ip4}}
    \end{minipage}
    \hfill
    \begin{minipage}[b]{.24\textwidth}
        \centering
        \subfloat[\label{fig:ip4_before_1_left}{$deg(w_1)=3$ in $\inGraphInc{C_l}$}]
        {\includegraphics[width=\textwidth,page=6]{images/ip4}}
    \end{minipage}
    \hfill
    \begin{minipage}[b]{.18\textwidth}
        \centering
        \subfloat[\label{fig:ip4_before_1_right}{Subgraph $\inGraphInc{C_r'}$}]
        {\includegraphics[width=\textwidth,page=7]{images/ip4}}
    \end{minipage}
    \hfill
    \begin{minipage}[b]{.50\textwidth}
        \centering
        \subfloat[\label{fig:ip4_after_1}{The resulting drawing of $\inGraphInc{C}$, when $deg(w_1)=3$ in $\inGraphInc{C_l}$}]
        {\includegraphics[width=\textwidth,page=8]{images/ip4}}
    \end{minipage}
    \hfill
    \begin{minipage}[b]{.24\textwidth}
        \centering
        \subfloat[\label{fig:ip4_before_2_left}{Subgraph $\inGraphInc{C_l'}$}]
        {\includegraphics[width=\textwidth,page=9]{images/ip4}}
    \end{minipage}
        \centering
        \begin{minipage}[b]{.18\textwidth}
        \subfloat[\label{fig:ip4_before_2_right}{Subgraph $\inGraphInc{C_r'}$}]
        {\includegraphics[width=\textwidth,page=10]{images/ip4}}
    \end{minipage}
    \hfill
    \begin{minipage}[b]{.50\textwidth}
        \centering
        \subfloat[\label{fig:ip4_after_2}{The resulting drawing of $\inGraphInc{C}$, when $deg(w_1)=2$ in $\inGraphInc{C_l}$}]
        {\includegraphics[width=\textwidth,page=11]{images/ip4}}
    \end{minipage}
    \caption{In all figures, $P(w_1\rightarrow w_j)$ is drawn fat, dotted edges are removed and gray-shaded dashed edges are added.}
    \label{fig:ip4}
\end{figure}

In the following, we describe how the two sub-instances
$\inGraphInc{C_l}$ and $\inGraphInc{C_r}$ can be recursively solved.
Observe that if $i\neq j$, then $C_l$ is not simple. In this case,
$C_l$ consists of a particular number of smaller simple subcycles,
for which IP-\ref{ip:3} and IP-\ref{ip:4} hold (hence they can be
recursively drawn), except for the first one, that is leftmost drawn
along $\ell$. First consider $\inGraphInc{C_r}$. We distinguish two
cases:

\begin{itemize}
\item \textbf{Case 1:} $j=m$ (see Fig.~\ref{fig:ip4_right_1}).
Then, $C_r$ is formed by $P(w_1\rightarrow w_m)$ and $(w_1,w_m)$.
Observe that $w_m$ is the rightmost vertex of $C_r$ and incident to
a chord of $C_s$. Hence, $deg(w_m)=2$ in $\inGraphInc{C_r}$. Since
none of the edges of $P(w_1\rightarrow w_m)$ is nested by a chord of
$C_r$, all vertices of $C_r$ (except possibly for $w_1$) are of
degree 2 in $\inGraphInc{C_r}$. If $deg(w_1)=3$ in
$\inGraphInc{C_r}$, then $(w_1,v)$ is bridge; a contradiction since
$G$ is biconnected. Hence, $\inGraphInc{C_r} \equiv C_r$. So, we
draw it as in Fig.~\ref{fig:ip4_right_1}, i.e., on the top-page.
Then, each subcycle of $C_l$ conforms to IP-\ref{ip:3} and
IP-\ref{ip:4} (including the first one, that is leftmost drawn along
$\ell$) and can be recursively drawn. The drawing of
$\inGraphInc{C_s}$ is derived by plugging the drawing of the
subcycles of $C_l$ into the drawing of $\inGraphInc{C_r}$. Observe
that the combinatorial embedding is preserved.

\item \textbf{Case 2:} $j<m$ (see
Fig.~\ref{fig:ip4_split_right}). All vertices of $P(w_1\rightarrow
w_j)$ have degree 2 in $\inGraphInc{C_r}$, except for $w_1$ and
$w_j$, that can have max-degree 3. We modify $\inGraphInc{C_r}$ as
follows (see Fig.~\ref{fig:ip4_right_2}): We contract
$P(w_1\rightarrow w_j)$ into a vertex, identified by $w_j$. Let
$\inGraphInc{C_r'}$ be the new subgraph with outerface $C_r'$.
Clearly, IP-\ref{ip:4} holds for $\inGraphInc{C_r'}$. IP-\ref{ip:3}
also holds, since $w_m$ is the rightmost vertex of $C_r'$ and
$deg(w_m)\leq3$. Hence, $\inGraphInc{C_r'}$ can be recursively
drawn. We proceed by distinguishing two sub-cases based on the
degree of $w_1$ in $\inGraphInc{C_l}$:

\begin{itemize}
\item \textbf{Case 2.1:} $deg(w_1)=3$ in $\inGraphInc{C_l}$
(see Fig.~\ref{fig:ip4_before_1_left}). Here IP-\ref{ip:3} and
IP-\ref{ip:4} hold for $\inGraphInc{C_l}$, so, it can be recursively
drawn. Once $\inGraphInc{C_l}$ and $\inGraphInc{C_r'}$ are drawn,
the drawing of $\inGraphInc{C}$ can be derived by deleting
$(w_j,w_m)$ from $\inGraphInc{C_r'}$ and restoring $(w_1,w_m)$, as
in Fig.~\ref{fig:ip4_after_1}. Since $w_1$ has no neighbors in
$\inGraph{C_r'}$, the embedding is preserved.

\item \textbf{Case 2.2:} $deg(w_1)=2$ in $\inGraphInc{C_l}$
(see Fig.~\ref{fig:ip4_before_2_right}): In this case, $deg(w_1)=3$
in $\inGraphInc{C_r'}$; see Fig.~\ref{fig:ip4_before_2_right}.
Again, we modify $\inGraphInc{C_l}$ as follows; see
Fig.~\ref{fig:ip4_before_2_left}. We remove $w_1$ and join $w_2$ and
$w_i$ by an edge. Let $\inGraphInc{C_l'}$ be the new subgraph with
outerface $C_l'$. IP-\ref{ip:4} may not hold for
$\inGraphInc{C_l'}$. However, $\inGraphInc{C_l'}$ has fewer vertices
than $\inGraphInc{C}$. We can benefit from this by proceeding
recursively, as we initially did with $\inGraphInc{C}$. Eventually,
at some point IP-\ref{ip:4} should hold, otherwise a graph with at
most $3$ vertices on its outerface should have a chord;
contradiction. Once $\inGraphInc{C_l'}$ has been drawn, we derive
the drawing of $\inGraphInc{C}$ as follows; see
Fig.~\ref{fig:ip4_after_2}. We remove $(w_2,w_i)$ and connect the
neighbors of $w_j$ in $\inGraphInc{C_r'}$ with its copy in
$\inGraphInc{C_l'}$. Note that no crossings are introduced, since
the two copies of $w_j$ in $\inGraphInc{C_l'}$ and
$\inGraphInc{C_r'}$ are consecutive on $\ell$. To complete the
drawing of $\inGraphInc{C}$, it remains to replace the copy of $w_j$
in $\inGraphInc{C_r'}$ with $w_1$, and add $(w_1,w_2)$ and
$(w_1,w_i)$.
\end{itemize}
\end{itemize}

To complete the description of our algorithm, it remains to describe
how the recursion begins. To do so, we need the following theorem,
that describes a simple property of planar graph drawing.

\begin{theorem}\label{thm:chordless}
Any planar graph $G$ admits a planar drawing $\Gamma(G)$ with a
chordless outerface.
\end{theorem}
\begin{proof}
Suppose that we are given a planar drawing $\Gamma'(G)$ of $G$, in
which the cycle, say $C: u_1 \rightarrow \ldots \rightarrow u_k
\rightarrow u_1$, bounding its outerface contains at least one
chord. Then, the endpoints of any chord of $C$ is a separation pair
of $G$. Let $(u_i,u_j)$, $1 \leq i < j \leq k$, be a chord of $C$
s.t. cycle $C': u_i \rightarrow u_{i+1} \rightarrow \ldots
\rightarrow u_j \rightarrow u_i$ has no chords. Let also $G_1$ and
$G_2$ be the two subgraphs of $G$ with outerfaces $u_i \rightarrow
u_{i+1} \rightarrow \ldots \rightarrow u_{j-1} \rightarrow u_j
\rightarrow u_i$ and $u_i \rightarrow u_{i-1} \rightarrow \ldots
\rightarrow u_{j+1} \rightarrow u_j\rightarrow u_i$ resp. Denote by
$f$ the face of $G_1$ that contains edge $(u_i,u_j)$. Since $u_i$
and $u_j$ is a separation pair of $G$, there exist a planar drawing
$\Gamma(G)$ of $G$ in which $G_2-\{(u_i,u_j)\}$ is drawn in the
interior of $f$ and the outerface of $G$ is bounded by the
chordless cycle $C'$.
\end{proof}

We are now ready to describe how the recursion begins. This is done
by specifying a drawing of $G$ with a chordless outerface, say
$C_{out}:v_1 \rightarrow \ldots v_k \rightarrow v_1$, which by
Theorem~\ref{thm:chordless} exists. Then, we place $v_1,\ldots,v_k$
in this order along $\ell$ and draw the edges of $C_{out}$ as
imposed by IP-\ref{ip:2}. If there is a vertex of $C_{out}$ with
degree less than four, then it is chosen as $v_k$ and all invariant
properties of our algorithm are satisfied. However, in the case
where such a vertex does not exist, it follows that $deg(v_k)=4$ in
$\inGraphInc{C_{out}}$ and therefor IP-\ref{ip:3} does not holds.

To cope with the latter case, we assume that we have computed the
\bvs of $\inGraph{C_{out}}$. Let $v_l$ ($v_r$, resp.) be the left
(right, resp.) neighbor of $v_k$ in $\inGraph{C_{out}}$ (see
Fig.~\ref{fig:outerface_before}) and $c_r$ ($c_l$, resp.) the \bv
that $v_r$ ($v_l$, resp.) belongs to ($c_r=c_l$ is possible).
Clearly, $v_l,v_r \notin C_{out}$, since $C_{out}$ is chordless. We
will augment $G$, s.t. IP-\ref{ip:3} holds in the augmented graph
$G_{aug}$. We outline our case analysis:

\begin{enumerate}[{Case~}1:]
\item\label{enum:1} $c_r$ is incident to a vertex of $C_{out}$ other than $v_k$.

\item\label{enum:2} $c_r$ is not incident to any other vertex of $C_{out}$ apart from $v_k$. In this case, once we define $G_{aug}$, we consider two additional subcases:

\begin{enumerate}[{Case~}\ref{enum:2}.1:]

\item Vertices $v_l$ and $v_r$ belong to two different \bvs of $G_{aug}$.

\item Vertices $v_l$ and $v_r$ belong to the same \bv of $G_{aug}$.

\end{enumerate}

\end{enumerate}

\begin{figure}[p]
    \begin{minipage}[b]{.44\textwidth}
        \centering
        \subfloat[\label{fig:outerface_before}{An instance in which IP-\ref{ip:3} is violated since $deg(v_k)=4$ in $\inGraphInc{C_{out}}$.}]
        {\includegraphics[width=.85\textwidth,page=1]{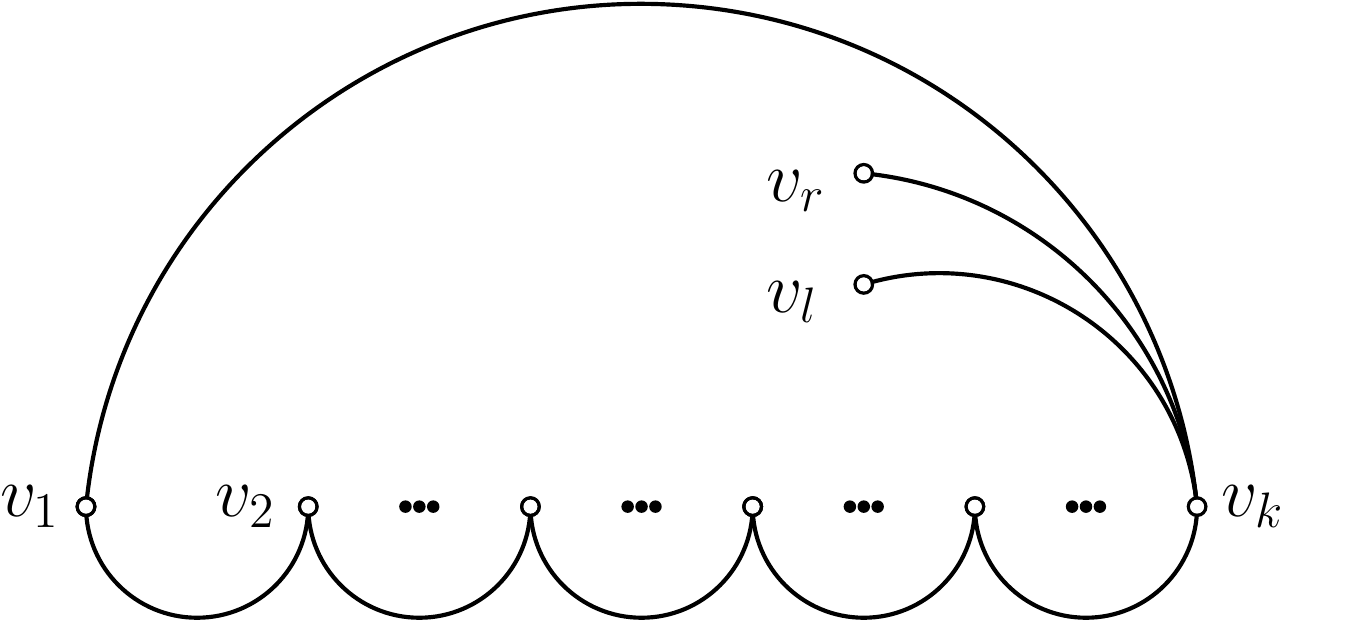}}
    \end{minipage}
    \hfill
    \begin{minipage}[b]{.44\textwidth}
        \centering
        \subfloat[\label{fig:outerface_before_4}{Graph $G_{aug}$, when $c_r$ is adjacent to a vertex of $C_{out}$ other than $v_k$.}]
        {\includegraphics[width=.85\textwidth,page=2]{images/ip3}}
    \end{minipage}
    \hfill
    \begin{minipage}[b]{.44\textwidth}
        \centering
        \subfloat[\label{fig:outerface_drawn_4}{The drawing of $G_{aug}$, when $c_r$ is adjacent to a vertex of $C_{out}$ other than $v_k$.}]
        {\includegraphics[width=.85\textwidth,page=3]{images/ip3}}
    \end{minipage}
    \hfill
    \begin{minipage}[b]{.44\textwidth}
        \centering
        \subfloat[\label{fig:outerface_restored_4}{A valid drawing of $G$, when $c_r$ is adjacent to a vertex of $C_{out}$ other than $v_k$.}]
        {\includegraphics[width=.85\textwidth,page=4]{images/ip3}}
    \end{minipage}
    \hfill
    \begin{minipage}[b]{.44\textwidth}
        \centering
        \subfloat[\label{fig:outerface_mirror}{Mirroring the input graph $G$ along the $y$--axis.}]
        {\includegraphics[width=.85\textwidth,page=5]{images/ip3}}
    \end{minipage}
    \hfill
    \begin{minipage}[b]{.44\textwidth}
        \centering
        \subfloat[\label{fig:outerface_modify}{Graph $G_{aug}$ when $c_r$ is not incident to a vertex of $C_{out}$ apart from $v_k$.}]
        {\includegraphics[width=.85\textwidth,page=6]{images/ip3}}
    \end{minipage}
    \hfill
    \begin{minipage}[b]{.44\textwidth}
        \centering
        \subfloat[\label{fig:outerface_drawn_1}{The drawing of $G_{aug}$, when: $c_r$ is not adjacent to another vertex of $C_{out}$ apart from $v_k$; $v_l$ and $v_r$ belong to different \bvs.}]
        {\includegraphics[width=.85\textwidth,page=7]{images/ip3}}
    \end{minipage}
    \hfill
    \begin{minipage}[b]{.44\textwidth}
        \centering
        \subfloat[\label{fig:outerface_restored_1}{A valid drawing of $G$, when: $c_r$ is not adjacent to another vertex of $C_{out}$ apart from $v_k$; $v_l$ and $v_r$ belong to different \bvs.}]
        {\includegraphics[width=.85\textwidth,page=8]{images/ip3}}
    \end{minipage}
    \hfill
    \begin{minipage}[b]{.44\textwidth}
        \centering
        \subfloat[\label{fig:outerface_shifted_2}{The drawing of $G_{aug}$, when: $c_r$ is not adjacent to any other vertex of $C_{out}$ apart from $v_k$; $v_l$ and $v_r$ belong to the same \bv; $(v_k,v_{k+1})$ is marked.}]
        {\includegraphics[width=.85\textwidth,page=9]{images/ip3}}
    \end{minipage}
    \hfill
    \begin{minipage}[b]{.44\textwidth}
        \centering
        \subfloat[\label{fig:outerface_restored_2}{A valid drawing of $G$, when: $c_r$ is not adjacent to any other vertex of $C_{out}$ apart from $v_k$; $v_l$ and $v_r$ belong to the same \bv; $(v_k,v_{k+1})$ is marked.}]
        {\includegraphics[width=.85\textwidth,page=10]{images/ip3}}
    \end{minipage}
    \hfill
    \begin{minipage}[b]{.44\textwidth}
        \centering
        \subfloat[\label{fig:outerface_general_3}{The drawing of $G_{aug}$, when: $c_r$ is not adjacent to any other vertex of $C_{out}$ apart from $v_k$; $v_l$ and $v_r$ belong to the same \bv; $(v_k,v_{k+1})$ is not marked.}]
        {\includegraphics[width=.85\textwidth,page=11]{images/ip3}}
    \end{minipage}
    \hfill
    \begin{minipage}[b]{.44\textwidth}
        \centering
        \subfloat[\label{fig:outerface_restored_3}{A valid drawing of $G$, when: $c_r$ is not adjacent to any other vertex of $C_{out}$ apart from $v_k$; $v_l$ and $v_r$ belong to the same \bv; $(v_k,v_{k+1})$ is not marked.}]
        {\includegraphics[width=.85\textwidth,page=12]{images/ip3}}
    \end{minipage}
        \caption{In all figures, dotted edges are removed and gray-shaded dashed edges are added.}
    \label{fig:ip3_all}
\end{figure}

\begin{itemize}
\item \textbf{Case~1:} Let $v_i$, $i<k$, be the
leftmost neighbor of $c_r$ on $C_{out}$. We augment $G$ as in
Fig.~\ref{fig:outerface_before_4}, by introducing three vertices to
the right of $v_k$. Let $C_{aug}$ be the outerface of the augmented
graph. Now observe that $G_{aug}$ satisfies IP-\ref{ip:3} and can be
recursively drawn. We claim that, in the drawing of $G_{aug}$, $v_r$
and $v_l$ are to the left of $v_k$, as in
Fig.~{\ref{fig:outerface_drawn_4}}. Denote by $c_r^{aug}$
($c_l^{aug}$, resp.) the \bv that $v_r$ ($v_l$, resp.) belongs to in
$G_{aug}-C_{aug}$. Note that $c_r^{aug}=c_l^{aug}$ is possible.
$c_r^{aug}$ is incident to $v_i$ through a marked edge, since $v_i$
is the leftmost neighbor of $c_r$. This implies that $c_r^{aug}$ is
placed directly next to $v_i$ (hence, to the left of $v_k$). Now,
observe that $v_k$ is the rightmost neighbor of $c_l^{aug}$. So,
$c_l^{aug}$ is placed to the left of $v_k$, even if $(v_l,v_k)$ is
marked, due to chord $(v_k,v_{k+2})$. Between $v_k$ and $v_{k+3}$ no
vertices of $G_{aug}$ exist, except for $v_{k+1}$ and $v_{k+2}$,
since the only \an that could be between $v_{k+2}$ and $v_{k+3}$ is
$c_r^{aug}$, which, however, is to the left of $v_k$, and so all
vertices of $G_{aug}-C_{aug}$ are to the left of $v_k$. If we
contract $v_k$, $v_{k+1}$, $v_{k+2}$ and $v_{k+3}$ back into $v_k$,
we obtain a valid drawing of $G$ (see
Fig.~\ref{fig:outerface_restored_4}).

\item \textbf{Case~2:} $c_r$ is not incident to any other vertex of
$C_{out}$ apart from $v_k$. We claim that we are allowed to assume
w.l.o.g. that $c_l$ is not incident to any other vertex of $C_{out}$
apart from $v_k$. If not so, consider a mirroring of $\ell$ at the
$y$--axis (see Fig.~\ref{fig:outerface_mirror}). The clockwise order
of the edges around each vertex of $G$ is reversed. So, $(v_k,v_l)$
($(v_k,v_r)$, resp.) is the right (left, resp.) edge of $v_k$. If
$c_l$ is incident to a vertex of $C_{out}$ other than $v_k$, then
Case~1 applies. Assume w.l.o.g. that $c_r$ and $c_l$ are not
incident to any other vertex of $C_{out}$ apart from $v_k$. We
augment $G$ as in Fig.~\ref{fig:outerface_modify}, s.t.
IP-\ref{ip:3} holds. Let $C_{aug}$ be the outerface of the augmented
graph. Since $v_k$ is not a cutvertex in $G$, $(v_k,v_{k+1})$ cannot
be a bridge in $G_{aug}$. Hence, $G_{aug}$ can be recursively drawn.
We distinguish two subcases:

\begin{itemize}
\item \textbf{Case~\ref{enum:2}.1:}  \emph{Vertices $v_l$ and
$v_r$ belong to two different \bvs of $G_{aug}-C_{aug}$}, say
$c_l^{aug}$ and $c_r^{aug}$ resp. So, $v_{k+1}$ belongs to another
\bv (containing only $v_{k+1}$) and is incident to $v_k$. Both
$c_r^{aug}$ and $c_l^{aug}$ are \ccs. Since $v_{k+1}$ is adjacent to
one vertex of $C_{out}$ (i.e. $v_k$), $(v_k,v_{k+1})$ is the marked
edge of $v_{k+1}$ and $v_{k+1}$ is placed directly to the left of
$v_k$ (see Fig.~\ref{fig:outerface_drawn_1}). $(v_k,v_{k+1})$ is
drawn on the bottom half-plane (marked edge) and $(v_l,v_{k+1})$ and
$(v_r,v_{k+1})$ are drawn on the top half-plane. If there was an \an
between $v_{k+1}$ and $v_k$, it would be adjacent to $v_k$,
contradicting the fact that $deg(v_k)=3$ in $G_{aug}$. So, the
rightmost \an of $G_{aug}-C_{aug}$ is $v_{k+1}$. Then, all vertices
of $G_{aug}-C_{aug}$ are to the left of $v_{k+1}$. So, if we
contract $v_k$ and $v_{k+1}$ back to $v_k$, then we obtain a valid
drawing of $G$ (see Fig.~\ref{fig:outerface_restored_1}).

\item \textbf{Case~\ref{enum:2}.2:} \emph{Vertices $v_l$ and $v_r$ belong to
the same \bv, say $c$, of $G_{aug}$}. Then, $v_{k+1}$ must belong to
$c$, as well. Also, $v_r$, $v_{k+1}$ and $v_l$ appear in this order
in the clockwise traversal of the outerface  $C_c$ of $c$. Since $c$
contains $v_{k+1}$, $c$ is adjacent to $v_k$ of $C_{out}$, and so
$c$ is incident to a marked edge, which ``determines'' the placement
of the vertices of $C_c$ on $\ell$. Let $v'$ be the vertex of $C_c$
incident to the marked edge of $c$. Since $c$ is adjacent to $v_k$,
$v'=v_{k+1}$ is possible (but $v'\notin\left\{ v_r, v_l\right\}$
since $v_r$ and $v_l$ are not incident to a vertex of $C_{out}$).

\begin{itemize}

\item Assume that $v'=v_{k+1}$, i.e.,
$(v_k,v_{k+1})$ is the marked edge of $c$ (see
Fig.~\ref{fig:outerface_shifted_2}). Then, $c$ is directly to the
left of $v_k$, with $v_{k+1}$ being the rightmost vertex of $C_c$.
Then between $v_{k+1}$ and $v_k$ no vertices of $G_{aug}$ exist,
since $deg(v_k)=3$ in $G_{aug}$, i.e., if there was an \an between
$v_k$ and $v_{k+1}$, it would be adjacent to $v_k$ and then
$deg(v_k)=4$. So, the rightmost \an of $G_{aug}-C_{aug}$ has
$v_{k+1}$ as its rightmost vertex. Then all vertices of
$G_{aug}-C_{aug}$ are to the left of $v_{k+1}$. If we contract
vertices $v_k$ and $v_{k+1}$ back to $v_k$, and draw $(v_l,v_k)$ on
the bottom half-plane and $(v_r,v_k)$ on the top half-plane, we
obtain a valid drawing of $G$ (see
Fig.~\ref{fig:outerface_restored_2}).

\item Assume now that $v'\neq v_{k+1}$. We claim that $v_r$,
$v_{k+1}$ and $v_l$ appear in this order from left to right on
$\ell$. Assume to the contrary that, either $v_l$ and $v_{k+1}$, or,
$v_{k+1}$ and $v_r$, are the leftmost and rightmost vertices of
$C_c$ on $\ell$, resp. The contradiction is implied by the
construction, in which $v'$ is either leftmost or rightmost on
$C_c$, and $v'\notin\left\{ v_{k+1}, v_r, v_l\right\}$. The current
situation is depicted in Fig.~\ref{fig:outerface_general_3}. If we
remove $v_{k+1}$, and draw $(v_l,v_k)$ and $(v_r,v_k)$ on the top
half-plane, then we obtain a valid drawing of $G$ (see
Fig.~\ref{fig:outerface_restored_3}).
\end{itemize}
\end{itemize}
\end{itemize}

We are now ready to state our main theorem.

\begin{theorem}
\label{thm:4degree} Any planar graph of maximum degree $4$ on $n$
vertices admits a two-page book embedding, which can be constructed
in $O(n^2)$ time.
\end{theorem}
\begin{proof}
At each step, our algorithm performs a series of computations; the
computation of the bridge-blocks, the topological sorting of
$G_{aux}^T$, BFS-traversals on the tangency trees. Using standard
algorithms from the literature all of these computations can be done
in $O(n)$ time, resulting in $O(n^2)$ total time.
\end{proof}

\section{Conclusions and Open Problems}
\label{sec:conclusions}
Two approaches were proposed to embed a 4-planar graph into two
pages. One reasonable question arising at this point is whether the
result can be extended to 5-planar graphs.

\bibliographystyle{plain}
\bibliography{references}

\end{document}